\documentclass[11pt,onecolumn,draftcls]{IEEEtran}
\usepackage{cite}
\usepackage{citesort}
\usepackage{graphicx}
\usepackage{amsmath}
\usepackage{times}
\usepackage{subfigure}
\usepackage{latexsym}
\usepackage{graphicx}
\usepackage{bm}
\usepackage{amssymb}
\usepackage[center]{caption2}
\usepackage{stfloats}
\usepackage{cases}
\usepackage{array}
\usepackage{setspace}
\usepackage{fancyhdr}
\usepackage{citesort}

\newtheorem{theorem}{Theorem}

\newtheorem{lemma}{Lemma}

\newcommand{\defeq}{\stackrel{\Delta}{=}}

\newcaptionstyle{mystyle2}{%
\captionlabel $.$ \, \doublespacing \captiontext \par}
\captionstyle{mystyle2}

\begin{document}
\title{Eigen-Based Transceivers for the \\
MIMO Broadcast Channel with \\
Semi-Orthogonal User Selection}
\author{Liang Sun, \emph{Student Member}, \emph{IEEE} and Matthew R.\ McKay, \emph{Member}, \emph{IEEE}
\thanks{Copyright (c) 2010 IEEE. Personal use of this material is permitted.
However, permission to use this material for any other purposes must
be obtained from the IEEE by sending a request to
pubs-permissions@ieee.org.} \thanks{Manuscript received Oct.\ 27,
2009; revised Feb.\ 10, 2010. The associate editor coordinating the
review of this manuscript and approving it for publication was Dr.
Ali Ghrayeb. L.\ Sun and M. R.\ McKay are with the ECE Department,
Hong Kong University of Science and Technology, Hong Kong. (Email:
sunliang@ust.hk; eemckay@ust.hk). The work of L. Sun and M. R.\
McKay was supported by the Hong Kong Research Grants Council (RGC)
under grant no.\ 617108. This work was presented in part at the IEEE
Global Communications Conference (Globecom), Honolulu, USA,
December, 2009.} } \maketitle

\begin{abstract}
This paper studies the sum rate performance of two low complexity
eigenmode-based transmission techniques for the MIMO broadcast
channel, employing greedy semi-orthogonal user selection (SUS). The
first approach, termed ZFDPC-SUS, is based on zero-forcing dirty
paper coding; the second approach, termed ZFBF-SUS, is based on
zero-forcing beamforming. We first employ new analytical methods to
prove that as the number of users $K$ grows large, the ZFDPC-SUS
approach can achieve the optimal sum rate scaling of the MIMO
broadcast channel. We also prove that the average sum rates of both
techniques converge to the average sum capacity of the MIMO
broadcast channel for large $K$. In addition to the asymptotic
analysis, we investigate the sum rates achieved by ZFDPC-SUS and
ZFBF-SUS for finite $K$, and show that ZFDPC-SUS has significant
performance advantages.  Our results also provide key insights into
the benefit of multiple receive antennas, and the effect of the SUS
algorithm. In particular, we show that whilst multiple receive
antennas only improves the asymptotic sum rate scaling via the
second-order behavior of the multi-user diversity gain; for finite
$K$, the benefit can be very significant. We also show the
interesting result that the semi-orthogonality constraint imposed by
SUS, whilst facilitating a very low complexity user selection
procedure, asymptotically does not reduce the multi-user diversity
gain in either first $\left(\log K\right)$ or second-order
$\left(\log\log K\right)$ terms.
\end{abstract}

\newpage
\section{Introduction}
In the multiple-input multiple-output (MIMO) broadcast channel, the
spatial multiplexing capability of multiple transmit antennas can be
exploited to efficiently serve multiple users simultaneously, rather
than trying to maximize the capacity of a single-user link. The
capacity region of the MIMO broadcast channel has now been
well-studied
\cite{W_Yu02,Viswanath03,Weingarten06,Vishwanath03,Caire03}, and has
been shown to be achieved through the use of multiple antenna dirty
paper coding (DPC) \cite{Weingarten06}.
Unfortunately, optimal DPC is a highly non-linear technique
involving joint optimization over a set of power-constrained
covariance matrices, and is therefore too complex for practical
implementation \cite{Vishwanath03}. A reduced complexity sub-optimal
DPC scheme, known as zero-forcing dirty paper coding (ZFDPC), was
proposed for single-antenna users in \cite{Caire03}, and generalized
to multiple-antenna users in \cite{Love07}, which is based on a QR
decomposition of the channel matrix.

To further reduce complexity, linear processing schemes such as
beamforming (BF) have also attracted a lot of attention.
The zero-forcing beamforming (ZFBF) scheme was first introduced for
single-antenna users in \cite{Caire03}, and further modified in
\cite{Peel05} and \cite{Hochwald05}. In \cite{Spencer04}, the
concept of block-diagonalization was proposed for multiple-antenna
users, which completely cancels the inter-user interference by
employing a set of precoding matrices.
One key limitation of these techniques is that, for ZFDPC and ZFBF,
the maximum number of users that can be supported must be no more
than the number of transmit antennas, whereas for
block-diagonalization, the number of the transmit antennas must be
larger than the aggregate number of receive antennas across all
users. This is significant, since the number of users in practice
can be large.

When the number of users $K$ is larger than the number of transmit
antennas $M$, one must select a subset of users in the system. A
common approach is to seek the subset of users which yields the
maximum sum rate. The complexity of finding the optimal subset,
however, can be prohibitively large, and to reduce complexity greedy
algorithms are commonly employed (see e.g.,
\cite{ZTu03,Dimic05,Shen06}). A promising way to further reduce the
complexity of user selection is to restrict the searching space of
users by imposing some constraint on the channels of the selected
users. Following this method, \cite{Yoo06} proposed a
semi-orthogonal user selection (SUS) algorithm which iteratively
searches for users with nearly orthogonal channel
directions\footnote{More specifically, two complex vectors
$\mathbf{u}$ and $ \mathbf{v}$, with unit norm, are said to be
semi-orthogonal if $|\mathbf{u}^H \mathbf{v}|^2< \delta$, where
$\delta$ is referred to as the \emph{semi-orthogonality
parameter.}}.


In this paper, we consider low complexity transmission and user
selection techniques for the MIMO broadcast channel with
multiple-antenna users. It is still not clear how much advantage can
be gained by employing multiple-antennas at the user terminals.
%
Some recent exceptions which deal with the multiple-antenna user
scenario are presented in \cite{Mohammad08} and \cite{Alireza08}.
Particularly, \cite{Mohammad08} proposed a generalized G-ZFDPC
approach, based on the idea of eigenmode transmission
(eigen-beamforming). 
A limitation of that approach is the relatively high complexity,
since it requires numerical optimization of certain system
parameters.
In \cite{Alireza08}, a thresholding technique based on the channel
singular values was proposed, and necessary and sufficient
conditions were given to achieve the optimum sum capacity of DPC as
$K \rightarrow \infty$. However, for that scheme, the optimal
threshold must be computed by exhaustive search, and is once again
quite complicated when the number of users is not small.


In this paper, we investigate two low complexity
eigen-beamforming-based transceiver structures for the MIMO
broadcast channel with multiple-antenna users, combined with a
greedy SUS algorithm. The first technique is a generalization the
G-ZFDPC approach in \cite{ZTu03} to account for multiple-antenna
users and combine it with SUS. We refer to this technique as
ZFDPC-SUS. The second technique is a generalization of the algorithm
proposed in \cite{Yoo06}, which we refer to as ZFBF-SUS. For both
techniques, we present an asymptotic performance analysis of the sum
rate (as in
\cite{Sharif05,Sharif07,Mohammad08,Alireza08,Love07,Yoo06}) as the
number of users grows large. In particular, by employing novel
analytical techniques, we demonstrate that ZFDPC-SUS achieves the
optimal sum capacity scaling of the MIMO broadcast channel as the
number of users grows large. In addition, we prove the more powerful
result that the difference between the sum rate of ZFDPC-SUS and the
sum capacity of the MIMO broadcast channel converges to zero. We
also establish a similar result for ZFBF-SUS. In addition to the
asymptotic analysis, we also investigate the sum rates achieved by
ZFDPC-SUS and ZFBF-SUS for finite $K$, for high and low
signal-to-noise ratios (SNR). Based on our analytical results, we
establish a number of important insights.  For example, we
demonstrate that by employing multiple-antennas at the user
terminals only affects the asymptotic sum rate scaling via the
second-order behavior of the multi-user diversity gain. Thus, the
improvement due to having multiple receive antennas at the terminals
is much less than that of having multiple transmit antennas, which
provides linear capacity growth through spatial multiplexing gain.
However, for finite $K$, we show that the performance improvement
due to multiple receive antennas can still be very significant. We
also establish key insights into the design of the
semi-orthogonality parameter used in the SUS algorithm. In
particular, it has been claimed previously that the
semi-orthogonality constraint will cause multi-user diversity gain
reduction \cite{Yoo06}. However, through our asymptotic analysis, we
show that if some very mild conditions on the semi-orthogonality
constraint are met, then the semi-orthogonality parameter \emph{does
not} reduce the multi-user diversity gain in either first or second
order, for both ZFDPC-SUS and ZFBF-SUS. It seems that this
conclusion cannot be established by using previous analytical
methods for SUS \cite{Yoo06}. Our analysis also leads to practical
design guidelines for selecting the semi-orthogonality parameter for
finite numbers of users, in order to intelligently trade off
complexity and performance. Our analysis also demonstrates that for
finite values of $K$, ZFDPC-SUS can significantly outperform
ZFBF-SUS.

\section{Channel and System Model}\label{sec:model}

We consider a MIMO broadcast channel with $M$ transmit antennas and $K$ users, with $K \geq M$. User $k$ is equipped
with $N_k$ antennas. In a flat-fading environment, the baseband
model of this system is
\begin{equation}\label{eq:channel_model}
\mathbf{y}_k = \mathbf{H}_k \mathbf{s} + \mathbf{n}_k, ~ 1 \leq
k\leq K ,
\end{equation}
where $\mathbf{y}_k \in \mathcal{C}^{N_k \times 1}$ is the received
signal vector of user $k$, $\mathbf{H}_k \in \mathcal{C}^{N_k \times
M}$ denotes the channel matrix from the transmitter to user $k$,
$\mathbf{s} \in \mathcal{C}^{M \times 1}$ represents the transmit
signal vector, designed to meet the total power constraint
$\text{Tr}( \mathcal{E}\{ \mathbf{s}\mathbf{s}^{H} \}) \leq P$, and
$\mathbf{n}_k \in \mathcal{C}^{N_k \times 1}$ is white Gaussian
noise with zero mean and covariance matrix $\mathbf{I}_{N_k}$.
Throughout the paper, we assume (as in
\cite{Caire03,Yoo06,Love08,Mohammad08}) that (i) the channels of all
users are subject to uncorrelated Rayleigh fading and, for
simplicity, all users are homogeneous and experience statistically
independent fading, (ii) the transmitter has perfect CSI of all
downlink channels\footnote{This assumption is reasonable in time
division duplex (TDD) systems, which allows the transmitter to
employ reciprocity to estimate the downlink channels.}, and (iii)
each user only has access to their own CSI, but not the
CSI of the downlink channels of the other users.

The transmitter supports $L \leq M$ simultaneous data streams,
shared by at most $L$ selected users
(active users), which are indexed by $\pi(i),~i=1,2,\cdots, L$.
(Note that the specific user selection algorithm will be discussed
in Section \ref{sec:algorithm}.) The transmitted signal vector is
represented as
\begin{equation} \label{eq:transmitSignal}
\mathbf{s} = \mathbf{W} \mathbf{P}^{\frac{1}{2}} \mathbf{x} ,
\end{equation}
where $\mathbf{x}=[x_1, x_2, \cdots, x_L]^T$ collects the zero-mean
circularly symmetric complex Gaussian information signals for each
of the $L$ data streams, satisfying $\mathcal{E}\{ \mathbf{x}
\mathbf{x}^H\} = \mathbf{I}_{L}$, $\mathbf{P} = \text{diag} \{ p_1,
p_2, \cdots, p_L \}$ accounts for the power loading across the
multiple streams, chosen to satisfy $\sum_{i=1}^{L}p_i \leq P$, and
$\mathbf{W} = [\mathbf{w}_1,\mathbf{w}_2,\cdots, \mathbf{w}_L] \in
\mathcal{C}^{M \times L}$ represents the precoder matrix, with
$\mathbf{w}_i $ denoting the beamforming vector for the $i$-th
stream (i.e.\ for user $\pi(i)$), normalized to satisfy $
\|\mathbf{w}_2\|^2 = 1$. Note that with this formulation, a given
user may be assigned multiple data streams.

From (\ref{eq:transmitSignal}), the received signal vector for user
$k$ can be rewritten as
\begin{equation}
\mathbf{y}_k = \mathbf{H}_k \mathbf{W}
\mathbf{P}^{\frac{1}{2}}\mathbf{x}+ \mathbf{n}_k.
\end{equation}
It is convenient to represent $\mathbf{H}_k$ via its singular value
decomposition (SVD) $\mathbf{H}_k = \mathbf{U}_k \mathbf{\Sigma}_k
\mathbf{V}_k^H $, where $\mathbf{\Sigma}_k$ is a $N_k \times M$
diagonal matrix containing the singular values of $\mathbf{H}_k$ in
decreasing order along its main diagonal, and $\mathbf {U}_k =
[\mathbf{u}_{k,1},\mathbf{u}_{k,2}, \cdots, \mathbf{u}_{k,N_k} ]\in
\mathcal{C}^{N_k \times N_k}$ and $\mathbf{V}_k
=[\mathbf{v}_{k,1},\mathbf{v}_{k,1},\cdots, \mathbf{v}_{k,M}] \in
C^{M \times M}$ are unitary matrices with $\mathbf{u}_{k,j}$ and
$\mathbf{v}_{k,j}$ representing the left and right singular vectors
corresponding to the $j$-th largest singular value
$\sqrt{\lambda_{k,j}}$.

To detect the data stream $i$, user $\pi(i)$ left multiplies the
received vector by 
$\mathbf{u}_{\pi(i),d_i}$ as follows
\begin{eqnarray} \label{eq:processedRx}
r_{\pi(i), d_i} &=& \mathbf{u}^H_{\pi(i),d_i} \mathbf{y}_{\pi(i)}
\nonumber\\&=& \sqrt{\lambda_{\pi (i), d_i}} ~\mathbf{v}^H_{\pi(i),
d_i} \mathbf{W} \mathbf{P}^{\frac{1}{2}}\mathbf{x} +
\tilde{n}_{\pi(i),d_i},~~
\end{eqnarray}
where $\tilde{n}_{\pi(i),d_i} = \mathbf{u}_{\pi(i),d_i}^H
\mathbf{n}_{\pi(i)} \sim \mathcal{CN}(0,1)$ is the effective
additive white Gaussian noise after processing, and $d_i$ denotes
the \emph{eigen-mode index} for stream $i$, chosen according to the
selection procedure outlined in Section \ref{sec:algorithm}.
Collecting the processed signals (\ref{eq:processedRx}) for each of
the $L$ data streams, we may write 
\begin{eqnarray}\label{eq:decoded_signal}
\mathbf{r} = \mathbf{C}_{\pi, d} \mathbf{W}
\mathbf{P}^{\frac{1}{2}}\mathbf{x} + \tilde{\mathbf{n}} =
\mathbf{\Lambda}_{\pi, d}^{\frac{1}{2}} \mathbf{\Xi}_{\pi, d}
\mathbf{W} \mathbf{P}^{\frac{1}{2}}\mathbf{x} + \tilde{\mathbf{n}} ,
\end{eqnarray}
where $\mathbf{C}_{\pi, d} = [\mathbf{c}_{\pi(1), d_1}^T
\mathbf{c}^T_{\pi(2), d_2}, ~\cdots, \mathbf{c}_{\pi(L), d_L}^T]^T$
is the composite channel matrix for the selected users and
eigen-channel set with $i$-th row vector $\mathbf{c}_{\pi(i), d_i} =
\sqrt{\lambda_{\pi (i), d_i}} \mathbf{v}_{\pi(i),d_i}^H$,
$\tilde{\mathbf{n}} = [~ \tilde{n}_{\pi(1), d_1}, \tilde{n}_{\pi(2),
d_2}, \cdots, \tilde{n}_{\pi(L), d_L}]^T $, $\mathbf{\Lambda}_{\pi,
d} = \text{diag}\{\lambda_{\pi (1), d_1}, \cdots, \lambda_{\pi (L),
d_L}\}$, and $\mathbf{\Xi}_{\pi, d} = [\mathbf{v}_{\pi(1),d_1},
\cdots, \mathbf{v}_{\pi(L),d_L} ]^H$.


In the next section, we will describe several transceiver
structures, as well as a greedy method for selecting the set of
active users $\pi =\{\pi(1), \cdots, \pi(L)\}$ and the corresponding
eigen-channels (active eigen-channels) $d = \{d_{1}, \cdots,
d_{L}\}$.


\section{Transceiver Structures and User Selection Algorithm}\label{sec:algorithm}
\subsection{Greedy Zero-Forcing Dirty Paper Coding Algorithm}\label{sec:GZFDPC}

In this subsection, we present a transmission strategy which jointly
combines ZF, DPC, and eigen-beamforming, along with a greedy low
complexity SUS scheduling algorithm. Henceforth, this strategy will
be termed ZFDPC-SUS. To the best of our knowledge this scheme has
not been considered before. We note, however, that it is an
extension of the ZFDPC strategy considered in
\cite{Caire03,ZTu03,Love08} to account for multiple receive
antennas, and also a variation of the algorithm discussed briefly in
\cite[Sect.\ VIII]{Yoo06}.




Let $\mathbf{\Xi}_{\pi, d} = \mathbf{L}_{\pi, d}\mathbf{Q}_{\pi, d}$
denote the QR decomposition of $\mathbf{\Xi}_{\pi, d}$, where
$\mathbf{L}_{\pi, d}$ is a $L\times L$ lower triangular matrix with
$(i,j)$-th entry $l_{i,j}$, and $\mathbf{Q}_{\pi,
d}=[\mathbf{q}_1^T, \cdots, \mathbf{q}_L^T]^T$ is a $L \times M$
matrix with orthonormal rows ($\mathbf{q}_i$ denotes the $i$-th row
vector). The transmit precoder matrix is chosen as
\begin{eqnarray}\label{eq:precoder}
\mathbf{W}= \mathbf{Q}_{\pi, d}^H.
\end{eqnarray}
Then, (\ref{eq:decoded_signal}) yields a set of interference
channels
\begin{equation}\label{eq:decoded_signal2}
r_{\pi(i), d_i} = \sqrt{\lambda_{\pi(i), d_i}} \big(\sqrt{p_i} ~
l_{i,i} x_i + \sum_{j<i} \sqrt{p_j}~ l_{i,j} x_j\big) +
\tilde{n}_{\pi(i), d_i}.
\end{equation}
From (\ref{eq:decoded_signal2}), if $i<j$, there is no interference
at receiver $\pi(i)$ from data stream $j$. For $i>j,$ the
interference term $\sum_{j<i} \sqrt{p_j}~ l_{i,j} x_j$ is
precanceled at the transmitter by using DPC. Then, the output SNR at
receiver $\pi(i)$ for data stream $i$ is given by
\begin{equation}\label{eq:SNR}
\zeta_{\pi (i), d_i} = p_i \gamma_{\pi(i), d_i}
\end{equation}
where $\gamma_{\pi(i), d_i} = \lambda_{\pi(i), d_i} \beta_i$, with
$\beta_i=|l_{i,i}|^2$.

Given the optimal user set $\pi$ and the corresponding eigen-channel set
$d$, the sum rate has the form
\begin{equation}\label{eq:throughtput}
R_{\text{ZFDPC-SUS}} = \max_{p_i: \sum_{i=1}^{L} p_i \leq P }
\sum_{i=1}^{L} \log_2 (1 + p_i \gamma_{\pi(i), d_i}).
\end{equation}
To maximize (\ref{eq:throughtput}), the power should be allocated
according to the standard water-filling algorithm.

Now consider the problem of selecting the optimal user set $\pi$ and
corresponding eigen-mode index set $d$.  These sets are chosen to maximize the sum rate, given by
(\ref{eq:throughtput}). When $M<K$, to find the optimal solution,
one must apply an \emph{exhaustive search} over all possible $L$,
and for each $L$, over all possible sets of $L$ subchannels taken
from the set of $\sum_{k=1}^{K}\text{min}\{M,N_k\}$ available
eigen-channels spanned by all $K$ users. Thus, the total number of
possible user and eigen-channel selection sets is given by
$\sum_{l=1}^{M}\binom{\sum_{k=1}^{K} \text{min}\{M,N_k\}}{l}$.
Further, since different orderings of a given set will yield
different output SNRs, all permutations of a given set must also be
considered.  Clearly, the complexity associated with this exhaustive
search is computationally prohibitive in practice, for all but small
values of $K$.

Here we consider a user and eigen-mode selection algorithm with
significantly lower complexity, based on SUS.  This algorithm, which
was first presented in \cite{Yoo06} in the context of ZFBF,
iteratively selects a user-eigenmode index pair by searching for a
set of users with near orthogonal channel vectors, and is described as follows.
Let $\mathcal{U}_n$ denote the \emph{candidate set} at the $n$-th
iteration. This set contains the indices of all users and the
corresponding eigen-channels that have not been selected previously,
and which have not been pruned in the previous iterations (i.e.,
they have satisfied the ``semi-orthogonality criteria'' in each of
the previous iterations). Also, let $\mathcal{S}_n= \{(\pi(1),d_1),
\cdots, (\pi(n),d_n)\}$ denote the set of indices of the selected
users and the corresponding eigen-channels after the $n$-th iteration.\\
\textbf{ZFDPC-SUS (Algorithm 1)}
\begin{enumerate}
\item
\textbf{Initialization:}\\
Set $n=1$ and $ \mathcal{U}_1 = \{(k,j)| ~k=1,2,\cdots,K; \, j=1,2,\cdots, \textrm{min}(N_k, M) \}$. \\
Let $\gamma_{k,j}(1) = \lambda_{k,j}$. The transmitter selects the first user and eigen-channel pair as follows:
\begin{eqnarray}\label{eq:selection1}
(\pi(1),d_1) = \textrm{arg} \max_{(k,j) \in \mathcal{U}_1}
~\gamma_{k,j}(1) \, .
\end{eqnarray}
Set 
$\mathcal{S}_1 = \{(\pi(1),d_1)\} $, and define $\mathbf{q}_1 =
\mathbf{v}_{\pi(1), d_1}^H$.
\item
\textbf{While} $n \leq M$, $n \leftarrow n+1$. \\
Calculate candidate set 
as
\begin{eqnarray}\label{set:semi-selection}
\mathcal{U}_{n} & =& \{ (k, j) | (k, j) \in \mathcal{U}_{n-1},
\nonumber
\\ && \hspace{-1cm}(k, j) \neq (\pi(n-1),d_{n-1}), | \mathbf{v}_{k,j}^H \;
\mathbf{q}_{n-1}^H|^2 < \delta \}\nonumber
\end{eqnarray}
where $\delta$ is a positive constant, termed the \emph{semi-orthogonality parameter}, that is preset before the
start of the selection procedure.

If $\mathcal{U}_n$ is empty,
set $n= n-1$ and go to step 3).
Otherwise, for each $(k,j) \in \mathcal{U}_n$, denote 
\begin{eqnarray}
\label{eq:Gram-Schmidt0}
\xi_{i} &=& \mathbf{v}_{k,j}^H\mathbf{q}_{i}^H,~~ i=1, \cdots, n-1\\
\label{eq:Gram-Schmidt1}
\boldsymbol{\xi}_{k,j} &=&
\mathbf{v}_{k,j}^H - \sum_{i=1}^{n-1}\xi_{i}\mathbf{q}_{i}\\
\label{eq:gamma_k} \gamma_{k,j}(n) &=&
\lambda_{k,j} \parallel\boldsymbol{\xi}_{k,j}\parallel^2.
\end{eqnarray}
Select the $n$-th active user and corresponding eigen-channel as
follows:
\begin{eqnarray}\label{eq:selection2}
\{(\pi(n), d_{n})\} &=& \text{arg} \max_{(k, j) \in \mathcal{U}_n}
\gamma_{k,j}(n) \, .
\end{eqnarray}
Set \begin{eqnarray} \mathcal{S}_n = \mathcal{S}_{n-1} \cup
{\{(\pi(n), d_{n})\}} , \nonumber \\ \mathbf{q}_n =
\frac{\boldsymbol{\xi}_{\pi(n), d_{n}}}{
\parallel\boldsymbol{\xi}_{\pi(n), d_{n}}\parallel}.  \label{eq:Gram-Schmidt2}
\end{eqnarray}
\item
The transmitter informs the selected users of the indices of their
selected eigen-channels; then performs DPC, beamforming, and
water-filling power allocation, as discussed previously.
\end{enumerate}
Note that this procedure applies Gram-Schmidt orthogonalization to the ordered rows of $\mathbf{\Xi}_{\pi, d}$, as
described by (\ref{eq:Gram-Schmidt0}), (\ref{eq:Gram-Schmidt1}) and
(\ref{eq:Gram-Schmidt2}).  As such, it also computes the required transmit precoding matrix in
(\ref{eq:precoder}).

Observe the following important relations. According to the
QR decomposition of $\mathbf{\Xi}_{\pi, d}$,
\begin{eqnarray}\label{eq:QR2}
\mathbf{v}_{\pi(n), d_n}^H = ~ (\mathbf{v}_{\pi(n), d_n}^H
\mathbf{q}_n^H )~\mathbf{q}_n + \sum_{j=1}^{n-1}(
\mathbf{v}_{\pi(n), d_n}^H \mathbf{q}_j^H ) \mathbf{q}_j,
\end{eqnarray}
and $l_{n,j} = \mathbf{v}_{\pi(n), d_n}^H \mathbf{q}_j^H$, for $
j<n$. With (\ref{eq:Gram-Schmidt1}),
\begin{eqnarray}\label{eq:beta_n}
\beta_n = |l_{n,n}|^2 = | \mathbf{v}_{\pi(n), d_n}^H \mathbf{q}_n^H
|^2 =
\parallel \boldsymbol{\xi}_{\pi(n), d_{n}} \parallel^2.
\end{eqnarray}
In addition, since $\|\mathbf{v}_{\pi(n), d_n}\|^2 = 1$ and
$\mathbf{q}_i, ~i=1,\cdots, L$ are orthonormal, it can be
easily shown that
\begin{eqnarray}\label{eq:sum_l}
\sum_{j=1}^{n} |l_{n,j}|^2 = 1, ~~~~\text{for}~ n=1,2,\cdots,L.
\end{eqnarray}



\subsection{Zero-Forcing Beamforming Algorithm}\label{sec:linear_processing}
The ZFDPC approach described in the previous section has
significantly lower complexity than full (capacity-achieving) DPC,
however it is still a nonlinear processing strategy, due to the
interference cancelation step. Thus, a common method for reducing
complexity even further is to remove the interference cancelation
and employ linear processing (linear beamforming). It is well-known,
however, that establishing the optimal linear beamforming vectors is
a very difficult non-convex optimization problem
\cite{H_Viswanathan03}. Instead, sub-optimal but simple linear
processing schemes are usually adopted. Here we will study ZFBF
which is one of the most popular linear strategies. Unless otherwise
indicated, we will employ the same notational symbols as used in the
previous sections.

Let $\mathbf{C}_{\pi, d}^{\dag}$ denote the Moore-Penrose inverse of
the equivalent channel matrix $\mathbf{C}_{\pi, d}$, i.e.,
$\mathbf{C}_{\pi, d}^{\dag} = \mathbf{C}_{\pi,d}^{H}
(\mathbf{C}_{\pi, d}\mathbf{C}_{\pi, d}^H)^{-1}$, and define ${\bf
\tilde{c}}_1, \ldots, {\bf \tilde{c}}_L$ as the columns of
$\mathbf{C}_{\pi, d}^{\dag}$. For ZFBF, the precoding matrix
$\mathbf{W} = [ \mathbf{w}_1, \ldots, \mathbf{w}_L ]$ is constructed
with the beamforming vectors $\mathbf{w}_i = \frac{{\bf {\tilde
c}}_i}{\parallel {\bf {\tilde c}}_i \parallel}$, for $i = 1, \ldots,
L$.
Note that this direct implementation of ZFBF requires the explicit
computation of the Moore-Penrose inverse of the channel matrix in
order to obtain the beamforming vectors. It has been shown in
\cite{Love08}, however, that this direct calculation can be
circumvented, thereby significantly reducing the computational
complexity. To this end, it is convenient to rewrite the
decomposition of $\mathbf{C}_{\pi, d}$ as $\mathbf{C}_{\pi, d} =
\mathbf{\Lambda}_{\pi, d}^{\frac{1}{2}} \mathbf{L}_{\pi, d}
\mathbf{Q}_{\pi, d}$, where $\mathbf{\Lambda} =
\text{diag}\{\lambda_{\pi(1), d_i},\cdots, \lambda_{\pi(L), d_L}\}$
and $\mathbf{L}_{\pi, d}, \mathbf{Q}_{\pi, d}$ are defined as in
Section \ref{sec:GZFDPC}.  Letting $\mathbf{T}_{\pi, d}=
\mathbf{L}_{\pi, d}^{-1} = [\mathbf{t}_1,\cdots,\mathbf{t}_L ] $,
assuming that $\mathbf{C}_{\pi, d}$ has full row rank, the
Moore-Penrose inverse $\mathbf{C}_{\pi, d}^\dagger$ can be written
as
\begin{align}
\mathbf{C}_{\pi, d}^\dagger =
\mathbf{Q}_{\pi, d}^H \mathbf{L}_{\pi, d}^{-1} \mathbf{\Lambda}_{\pi, d}^{-\frac{1}{2}} \; .
\end{align}
Note that calculating the inverse of $\mathbf{\Lambda}_{\pi,
d}^{\frac{1}{2}}$ is trivial (since it is diagonal), whereas the
inverse of $\mathbf{L}_{\pi, d}$ can be computed using a simple
iterative algorithm given in \cite[Eq. 11]{Love08}.

For ZFBF, the decoded signal for data stream $\pi(i)$ is easily shown to be given by
\begin{eqnarray}
r_{\pi(i), d_i} &=& \sqrt{p_i}~ \mathbf{c}_{\pi(i), d_i}
\mathbf{w}_{i} x_i
+ \tilde{n}_{\pi(i), d_i}\nonumber\\
&= & \frac{\sqrt{p_i~\lambda_{\pi(i), d_i}}}{\parallel \mathbf{t}_i
\parallel} x_i
+ \tilde{n}_{\pi(i), d_i}
\end{eqnarray}
with corresponding SNR
\begin{align}
\varrho_{\pi(i), d_i} = \frac{\lambda_{\pi(i),
d_i}}{\parallel \mathbf{t}_i \parallel^2} \; .
\end{align}

For the given user set $\pi$ and the corresponding eigen-channel set
$d$, the sum rate is given by
\begin{eqnarray}\label{eq:sum_rate_ZFBF}
R_{\text{ZFBF-SUS}}  = \max_{p_i: \sum_{i=1}^{L} p_i \leq P }
\sum_{i=1}^{L} \log_2(1+p_i \varrho_{\pi(i), d_i}) ,
\end{eqnarray}
where the optimal power allocation $\{p_i\}_{i=1}^L$ is obtained,
once again, by applying the waterfilling procedure.

For ZFBF, we consider a user and eigen-channel selection algorithm
based on SUS, following the same general procedure as in
\textbf{Algorithm 1}. Note that SUS has previously been applied to
ZFBF in \cite{Yoo06}. This algorithm typically assumes that each
user is equipped with a single receive antenna, however it extends
easily to the multiple receive antenna scenario considered in this
paper. One key difference between the algorithms in
\cite{Yoo06,Dimic05,Love08} are the specific methods employed for
selecting the ``best'' user in Step 2 of the algorithm. More
specifically, in \cite{Yoo06}, the same method was applied as in
(\ref{eq:selection2}), whereas \cite{Dimic05} applied a method based
on selecting one user at each iteration that results in the largest
sum rate when combined with previously selected users. Whilst the
latter method can result in larger sum rate, here we will consider
the former method for analytically tractability. It has been shown,
however, that the difference in sum rate between these two methods
is minor \cite{Love08}.

\section{Sum Rate Analysis -- Asymptotic $K$}\label{sec:analysis}
In this section, we investigate the average sum rate of each of the above transceiver
structures. For tractability, we make the following
assumptions throughout this section:
\begin{enumerate}
\item[(i)] For each user, only the principal
eigen-channel is considered. As such, we drop the indices for the selected eigen-channels  (for example,
we use $\gamma_{\pi(i)}$ instead of $\gamma_{\pi(i), d_i}$).

\item[(ii)] The available power $P$ is divided equally amongst
the active users\footnote{Note that in practice the transmit power
may be optimized (e.g., according to the water-filling strategy). In
such cases, the power allocation depends on the instantaneous
channel coefficients and thus changes at the fading rate of the
channel, which makes the analysis intractable.}.
\end{enumerate}
Clearly, the sum rate achieved under these two assumptions will
serve as a lower bound to the maximum achievable sum rate. We will
also assume that each user has $N$ antennas, and that there are $L =
M$ data streams.


We will investigate the average sum rate of both scheme discussed in the previous section.  We focus on
establishing asymptotic results as $K \to \infty$, whilst keeping SNR, $M$, and $N$ fixed.


\subsection{ZFDPC-SUS Scheme}\label{sec:perforance_GZFDPC}
To analyze the sum rate of the ZFDPC-SUS system, we require the
distribution of the output SNR $\zeta_{\pi(n)}$, or alternatively
the distribution of $\gamma_{\pi(n)}$. Let us first determine the
distribution of $\gamma_{k} (n)$, $n = 1, \cdots, M$, where $k$ is
an \emph{arbitrary} user selected from the candidate set
$\mathcal{U}_n$.

Starting with~$n=1$, $\gamma_k(1)$, $k = 1, \ldots, K$, are
independent and identically distributed (i.i.d.), with
\begin{eqnarray}
\gamma_{k} (1)= \lambda_{k,\text{max}}
\end{eqnarray}
where $\lambda_{k,\text{max}}$ is the maximum eigenvalue of
$\mathbf{H}_k^H \mathbf{H}_k$, whose probability density function
(p.d.f.) and cumulative distribution function (c.d.f.) are known in
closed-form and are given as follows\cite{Mallik03}:
\begin{lemma}\label{lemma:max_eigen_pdf}
Let $\mathbf{H} \sim \mathcal{CN}_{N,M} (\mathbf{0}_{N,M},
\mathbf{I}_{N}\otimes \mathbf{I}_{M})$. The matrix $\mathbf{H}^H
\mathbf{H}$ is complex Wishart, whose maximum eigenvalue has p.d.f.\
\begin{eqnarray}\label{eq:wishart_pdf}
f_{\text{max}} (x) = \sum_{r=1}^{p} \sum_{s=q-p}^{(p+q-2r)r}
a_{r,s}~ x^{s} e^{-r x}
\end{eqnarray}
and c.d.f.\
\begin{eqnarray}\label{eq:wishart_cdf}
F_{\text{max}} (x) = \sum_{r=1}^{p} \sum_{s=q-p}^{(p+q-2r)r}
\frac{a_{r,s}}{r^{s+1}} \gamma (s+1, r x)
\end{eqnarray}
where $ p = \min\{M,N\}$, $q = \max\{M,N \}$, $a_{s,r}$ is a
constant (dependent on $M$ and $N$) which can be computed using the
simple numerical method in \cite{Maaref05}, and $\gamma(\cdot,
\cdot)$ is the lower incomplete gamma function.
\end{lemma}

For $n\geq 2$, evaluating the distribution of $\gamma_{k}(n)$, $k
\in \mathcal{U}_n$, is significantly more challenging. Particularly,
the ``max'' operation (\ref{eq:selection1}) of Step 1 of the
previous iteration (i.e., the $\left(n-1\right)$-th), and also the
semi-orthogonality constraint imposed at Step 2 of the current
iteration (i.e., the $n$-th) will make the exact distribution of the
eigen-channel vectors in $\mathcal{U}_n$ different from the
distributions of the eigen-channel vectors in $\mathcal{U}_l$, $ l
\leq n-1$. More specifically, for $n\geq 2$, the eigen-channels for
users in the candidate set $\mathcal{U}_n$ are no longer distributed
according to the maximum eigen-channel of a complex Wishart matrix
(i.e., for $k \in \mathcal{U}_n$, $\mathbf{v}_k$  is no longer an
isotropically distributed unit vector on the complex unit sphere,
and $\lambda_{k,\text{max}}$ is no longer distributed as the maximum
eigenvalue of a complex Wishart matrix). 

We see from (\ref{eq:gamma_k}) that $\gamma_{k}(n)$ involves the
\emph{product} of $\lambda_{k,\text{max}}$ and the projection
variable $\parallel\boldsymbol{\xi}_{k}\parallel^2$. For the reasons
stated above, the exact distributions of both
$\lambda_{k,\text{max}}$ and
$\parallel\boldsymbol{\xi}_{k}\parallel^2$ for $k\in \mathcal{U}_n,
n\geq 2$ are currently unknown and appear very difficult to derive
analytically. Fortunately, we can make progress by appealing to the
``large-user'' regime.
In particular, when the number of users in the candidate set
$\mathcal{U}_n$ is large, the problem is greatly simplified by
invoking the following key lemma, which shows that removing a finite
number of users from $\mathcal{U}_n$ has negligible impact on the
statistical properties of the remaining users. Similar results have also
been established previously for different system configurations \cite{Yoo06,Dimic05,Love08}.

%


\begin{lemma}\label{lemma:iid_property}
At the $n$-th iteration, $ 2\leq n \leq M$, conditioned on the
previously selected eigen-channel vectors $\mathbf{c}_{\pi(1)},
\cdots, \mathbf{c}_{\pi(n-1)}$, the eigen-channel vectors in
$\mathcal{U}_n$ are i.i.d. Furthermore, as the size of the candidate
user set $\mathcal{U}_n$ grows large (i.e. $\lim_{K \rightarrow
\infty} |\mathcal{U}_n| = \infty$), conditioned on the previously
selected eigen-channels $\mathbf{c}_{\pi(1)}, \cdots,
\mathbf{c}_{\pi(n-1)}$, the eigen-channel for each user in
$\mathcal{U}_n$ converges in distribution to the distribution of the
principal eigen-channel of a complex Wishart matrix.
\end{lemma}
\begin{proof}
See Appendix \ref{proof_lemma:iid_property}.
\end{proof}


Note that our result here differs from that of \cite{Love08} in both
the distribution of the channel vectors and also the user selection
algorithm.

Equipped with \emph{Lemma} \ref{lemma:iid_property}, at the $n$-th
iteration, from the point of view of the users in $\mathcal{U}_n$,
the eigen-channel vectors of the selected users in the previous
iterations (i.e., $\mathbf{c}_{\pi(1)}, \cdots,
\mathbf{c}_{\pi(n-1)}$) appear to be \emph{randomly} selected. Thus,
the orthonormal basis $\mathbf{q}_1, \cdots, \mathbf{q}_{n-1}$
(generated from $\mathbf{c}_{\pi(1)}, \cdots,
\mathbf{c}_{\pi(n-1)}$) appears independent of the eigen-channel
vectors of the users in $\mathcal{U}_n$.
This greatly simplifies the following analysis.

We require the exact distribution of
$\gamma_{k}(n) = \lambda_{k,\text{max}}
\parallel\boldsymbol{\xi}_{k}\parallel^2$. To this end, the
major challenge is to derive the c.d.f.\ of $\beta_{k}(n)= \|
\boldsymbol{\xi}_{k} \|^2 $ for an arbitrary user $k\in
\mathcal{U}_n$, i.e. $F_{\beta (n)}(x) = \text{Pr}\big(\beta_{k} (n)
\leq x |~k \in \mathcal{U}_n \big)$. Recalling that $l_{n,j} =
\mathbf{v}_{\pi(n), d_n}^H \mathbf{q}_j^H$ for $ j<n$, with
(\ref{eq:beta_n}) and (\ref{eq:sum_l}), we can re-express this
c.d.f.\ as follows:
\begin{eqnarray}\label{eq:beta_dist}
F_{\beta (n)}(x) &=& \text{Pr}\left( |\mathbf{v}_k^H \mathbf{q}_n^H
|^2 \leq x ~\big|~|\mathbf{v}_k^H \mathbf{q}_1^H |^2< \delta,
\cdots,|\mathbf{v}_k^H \mathbf{q}_{n-1}^H |^2< \delta\right) \nonumber\\
&=&\text{Pr}\left( \sum_{i=1}^{n-1} |\mathbf{v}_k^H \mathbf{q}_i^H
|^2 \geq 1- x ~\bigg|~|\mathbf{v}_k^H \mathbf{q}_1^H |^2<
\delta,\cdots,
|\mathbf{v}_k^H \mathbf{q}_{n-1}^H |^2< \delta\right)\nonumber\\
& = & 1- \frac{\text{Pr}\left( \sum_{i=1}^{n-1} |\mathbf{v}_k^H
\mathbf{q}_i^H |^2 \leq 1- x , |\mathbf{v}_k^H \mathbf{q}_1^H |^2<
\delta,\cdots, |\mathbf{v}_k^H \mathbf{q}_{n-1}^H |^2 <
\delta\right)}{\text{Pr}\left( |\mathbf{v}_k^H \mathbf{q}_1^H |^2<
\delta,\cdots, |\mathbf{v}_k^H \mathbf{q}_{n-1}^H |^2< \delta
\right)}.
\end{eqnarray}

The denominator, $\mu_{n}(\delta) \defeq \text{Pr}\big(
|\mathbf{v}_k^H \mathbf{q}_1^H |^2 < \delta,\cdots, |\mathbf{v}_k^H
\mathbf{q}_{n-1}^H |^2 < \delta\big)$, 
denotes the probability that any arbitrary user $k \in \{ 1, \ldots,
K \}$ will belong to the set $\mathcal{U}_n$. Note that this
probability has also been considered in the context of ZFBF for the
MIMO broadcast channel in \cite{Yoo06}, where a rather loose lower
bound was derived. Here we derive an exact expression which applies
for large $K$, using an alternative derivation approach. For
tractability, our result applies for $\delta < \frac{1}{M-1}$, which
is easy to establish.
\begin{lemma}\label{lemma:distri_mu_delta}
With sufficiently large $K$ and $\delta < \frac{1}{M-1}$, the
probability that an arbitrary user $k \in \{ 1, \ldots, K \}$
belongs to the set $\mathcal{U}_{n}$, for $n \in \{ 2, \cdots, M
\}$, is given by
\begin{align}\label{eq:mu_delta}
\mu_{n}(\delta) &=  \text{Pr}\left( |\mathbf{v}_k^H \mathbf{q}_1^H
|^2<
\delta,\cdots, |\mathbf{v}_k^H \mathbf{q}_{n-1}^H |^2< \delta\right) \nonumber\\
&= \sum_{k=n-1}^{M-1} \binom{M-1}{k} (-1)^{k} \bigg[
\sum_{i=0}^{n-1} \binom{n-1}{i} (-1)^{i} i^k \bigg] \delta^k.
\end{align}
\end{lemma}
\begin{proof}
See Appendix \ref{proof_lemma:distri_mu_delta}.
\end{proof}
Note that the term ``sufficiently large'' in \emph{Lemma
\ref{lemma:distri_mu_delta}} implies that $K$ should be large enough
such that:
\begin{eqnarray}\label{eq:card_n}
\mathcal{K}_n = |\mathcal{U}_n| \approx K \mu_{n}(\delta)
\end{eqnarray}
due to the law of large numbers (LLN). In fact, this also places an
additional requirement on $\delta$, which must be selected such that
as $K \rightarrow  \infty$, $|\mathcal{U}_n|$ becomes sufficiently
large (e.g.\ such that $\lim_{K \rightarrow \infty} |\mathcal{U}_n|
= \infty$). More specifically, since $\delta < 1$, by examining
(\ref{eq:card_n}) and (\ref{eq:mu_delta}) and recalling the
condition on $\delta$ in the lemma statement, we can establish the
following design criterion: $\delta$ should be chosen such that
\begin{align} \label{eq:Conditions}
\lim_{K \to \infty} K \delta^{M-1} = \infty \; \; {\rm and} \; \;
\delta < \frac{1}{M-1} \; .
\end{align}
This implies that any $\delta$ can be selected, as long as it does
not approach zero at a rate of $1/K^{\frac{1}{M-1}}$ or faster as $K
\to \infty$, whilst also meeting the technical condition $\delta <
\frac{1}{M-1}$.  These are very mild conditions which are easy to
satisfy (for example, choosing $\delta$ to be any constant less than
$\frac{1}{M-1}$). We further discuss the design implications of
selecting $\delta$ in Section \ref{sec:Comparison}.


The numerator in (\ref{eq:beta_dist}) can
be evaluated using similar methods, which leads to the following result:
\begin{lemma} \label{lemma:distri_beta}
Let $ k \in \mathcal{U}_{n}$, $n \in \{ 2, \cdots, M \}$, and assume
$\delta$ is chosen to satisfy (\ref{eq:Conditions}). For
sufficiently large $K$, the c.d.f.\ of $\beta_k(n)$, given in
(\ref{eq:beta_dist}), can be expressed as follows:
\begin{eqnarray}\label{eq:distri_beta_n}
F_{\beta (n)}(x) &=&  \left\{
\begin{array}{lll}
0 \, ,  & x \leq 1-(n-1)\delta\\
1 - \frac{\Gamma(M)}{\Gamma(M-n+1)  \mu_{n}(\delta)}\\
\hspace{0.1cm}\times \int_{t_{n-1}}\cdots \int_{t_1}
(1-\sum_{i=1}^{n-1}t_i)^{M-n} {\rm d}t_1\cdots {\rm d}t_{n-1} \, ,
&   1- (n-1)\delta < x \leq 1\\
1 \, , & x > 1
\end{array}
\right.
\end{eqnarray}
where the integral region is given by $t_{i} \in \bigg[0,
\min\big\{\delta, 1-x - \sum_{j=i+1}^{n-1} t_j \big\}\bigg]$.

For $n=2$, (\ref{eq:distri_beta_n}) has the closed-form solution
\begin{align}\label{eq:distri_beta_2}
F_{\beta (2)}(x) = \left\{
\begin{array}{lll}
0 ~~ &  x \leq  1- \delta \\
\frac{x^{M-1} - (1-\delta)^{M-1}}{1- (1-\delta)^{M-1}} ~~ &   1- \delta < x \leq 1\\
1 ~~ & x > 1
\end{array} .
\right.
\end{align}
\end{lemma}
\begin{proof}
See Appendix \ref{proof_lemma_distri_beta}.
\end{proof}
For arbitrary $M$ and $n$, it is difficult to obtain an exact
closed-form solution for this c.d.f.\  Based on the above lemma,
however, we can derive closed-form \emph{upper and lower bounds}, as
given by the following:
\begin{lemma}\label{lemma:distri_beta_bound}
The c.d.f.\ $F_{\beta (n)}(x)$, for $n \in \{ 2, \cdots, M \}$,
satisfies $F_{\bar{\beta}(n)} (x) \leq F_{\beta (n)}(x) \leq
F_{\tilde{\beta} (n)}(x)$, with $F_{\tilde{\beta} (n)}(x)$ and
$F_{\bar{\beta}(n)} (x)$ given by (\ref{eq:dom_distri_beta_n}) and
(\ref{eq:domed_distri_beta_n})
\begin{eqnarray}\label{eq:dom_distri_beta_n}
F_{\tilde{\beta} (n)}(x) &=&  \left\{
\begin{array}{lll}
0 ~~ &   x  \leq  1- (n-1)\delta\\
1- \frac{\mu_{n}(\frac{1-x}{n-1})}{\mu_{n}(\delta)}  ~~ &  1-
(n-1)\delta  < x \leq 1 \\
1 ~~ &  x>1
\end{array}
\right.
\end{eqnarray}
and
\begin{eqnarray}
\label{eq:domed_distri_beta_n} F_{\bar{\beta} (n)}(x) &=& \left\{
\begin{array}{lll}
0 ~~ &   x  \leq  1- (n-1)\delta\\
1 - \frac{I_{1-x} (n-1, M-n+1)}{\mu_{n} (\delta)} ~~ &  1-
(n-1)\delta  < x \leq 1 \\
1 ~~ &  x>1
\end{array}
\right.
\end{eqnarray}
respectively, where $\mu_{n}(\cdot)$ is given by (\ref{eq:mu_delta})
and $I_{x}(\cdot, \cdot)$ is the regularized incomplete beta
function.

Note that for $n=2$, $F_{\beta_k(n)} (x) = F_{\bar{\beta}(n)} (x) = F_{\tilde{\beta}(n)} (x)$.
\end{lemma}
\begin{proof}
See Appendix \ref{proof_lemma:distri_beta_bound}.
\end{proof}

Equipped with \emph{Lemma \ref{lemma:distri_beta_bound}}, and with
the help of \emph{Lemma} \ref{lemma:max_eigen_pdf}, we may now
derive upper and lower bounds on the c.d.f.\ of $\gamma_{k}(n)$. To
establish this result, recall that for an arbitrary user $k \in
\mathcal{U}_n$, $n \geq 2$, then $\gamma_k(n) = \lambda_{k,
\text{max}} \beta_k(n)$. Also, define $\bar{\gamma}_k{(n)} =
\lambda_{k,\text{max}} \bar{\beta}_k(n)$ and $\tilde{\gamma}_k{(n)}
= \lambda_{k,\text{max}} \tilde{\beta}_k(n)$, with c.d.f.s
$F_{\bar{\gamma}(n)}(x)$ and $F_{\tilde{\gamma}(n)}(x)$
respectively.
\begin{lemma}\label{lemma:distri_DPC}
The c.d.f.\ $F_{\gamma(n)}(x)$, for $n \in \{ 2, \cdots, M \}$,
satisfies $F_{\bar{\gamma}(n)}(x) \leq F_{\gamma(n)}(x) \leq
F_{\tilde{\gamma}(n)}(x)$, with $F_{\tilde{\gamma}(n)}(x)$ and
$F_{\bar{\gamma}(n)}(x)$ given by
\begin{eqnarray}\label{eq:upper_distri_DPC}
F_{\tilde{\gamma}(n)}(x) &=&  F_{\text{max}} \left(\frac{x}{t}\right) -
\frac{1}{\mu_{n}(\delta)} \sum_{k=n-1}^{M-1} \binom{M-1}{k}(-1)^k
\left[ \sum_{i=0}^{n-1} \binom{n-1}{i}(-1)^i \left(\frac{i}{n-1} \right)^k
\right] \sum_{r=1}^{p} \sum_{s=q-p}^{(N+M-2r)r} a_{r,s} \nonumber \\
&& \hspace*{1cm} \times \sum_{j=0}^{k} \binom{k}{j} r^{k-j-s-1}
(-x)^{k-j} \left[\Gamma\left(j-k+s+1, rx \right) -
\Gamma\left(j-k+s+1, \frac{rx}{t}\right) \right].
\end{eqnarray}
\begin{eqnarray}
\label{eq:lower_distri_DPC} F_{\bar{\gamma}(n)}(x) &=&
F_{\text{max}} \left( \frac{x}{t} \right)- \frac{1}{\mu_n(\delta)}
\sum_{k=0}^{M-n} \binom{M-1}{k} (-1)^k
\sum_{r=1}^{p}\sum_{s=q-p}^{(N+M-2r)r} a_{r,s} \sum_{j=0}^{M-k-1} \binom{M-k-1}{j} r^{M-j-s-2} \nonumber\\
&& \hspace*{1cm} \times (-x)^{M-j-1} \left[ \Gamma\left(j+s-M+2, r
x\right)- \Gamma \left(j+s-M+2, \frac{r x}{t} \right) \right].
\end{eqnarray}
respectively, where $F_{\text{max}}(\cdot)$, $p$, $q$ and $a_{r,s}$
are defined as in \emph{Lemma} \ref{lemma:max_eigen_pdf},
$t=1-(n-1)\delta$ and $\Gamma(\cdot, \cdot)$ denotes the upper
incomplete gamma function.

For the case $n=2$, $F_{\gamma(n)}(x) = F_{\tilde{\gamma}(n)}(x)
=F_{\bar{\gamma}(n)}(x)$.
\end{lemma}

\begin{proof}
See Appendix \ref{proof_distri_DPC}.
\end{proof}
Although not shown due to space limitations, these bounds have been confirmed through simulations.

Recall that our primary aim is to characterize the distribution
of $\zeta_{\pi(n)}$, or equivalently $\gamma_{\pi(n)}$ which, from
(\ref{eq:selection2}), is the maximum of a collection of i.i.d.\
random variables chosen from $\mathcal{U}_n$, with common c.d.f.\
$F_{\gamma(n)} (x)$. Moreover, as discussed previously, our main
interest is the case where the number of users $K$, and consequently
the size of $\mathcal{U}_n$, is large.  As such, from the theory of
extreme order statistics (see e.g.\ \cite[Appendix
I]{Mohammad08}\cite{oreder_statistics03}),
the asymptotic distribution of the largest order statistic
$\gamma_{\pi(n)}$ depends on the \emph{tail} behavior (large $x$) of
$F_{\gamma(n)} (x)$. For $n \geq 2$, the following closed-form
asymptotic (high $x$) expansions for the c.d.f.\ upper and lower
bounds in (\ref{eq:upper_distri_DPC}) and
(\ref{eq:lower_distri_DPC}) are derived in Appendix
\ref{app:tail_gamman}:
\begin{eqnarray}
\label{eq:upper_tail_distri} F_{\tilde{\gamma}(n)}(x) &=& 1-
\frac{1}{\mu_n(\delta)~ \varepsilon_n} e^{-x} x^{M+N-n-1}
\nonumber\\ && +
~O(e^{-x}x^{M+N-n-2}) \\
\label{eq:lower_tail_distri} F_{\bar{\gamma}(n)}(x) &=& 1-
\frac{1}{\mu_n(\delta)~\epsilon_n} e^{-x} x^{M+N-n-1} \nonumber\\
&& +~ O(e^{-x}x^{M+N-n-2})
\end{eqnarray}
where
\begin{eqnarray}\label{eq:parameter}
\frac{1}{\varepsilon_n} & =& \frac{\Gamma(n)}{
\Gamma(M-n+1)\Gamma(N)(n-1)^{n-1}},\\  \frac{1}{\epsilon_n} & =&
\frac{1}{ \Gamma(M-n+1)\Gamma(N)}.
\end{eqnarray}

Based on the above results, we can establish upper and lower bounds
of the asymptotic distribution of $\gamma_{\pi(n)}$, for large $K$.
To this end, define $\tilde{\gamma}_{\pi(n)} = \max_{k \in
\mathcal{U}_n} \tilde{\gamma}_{k}(n)$ and $\bar{\gamma}_{\pi(n)} =
\max_{k \in \mathcal{U}_n} \bar{\gamma}_{k}(n)$, with c.d.f.s
$F_{\tilde{\gamma}_{\pi(n)}}(x)$ and $F_{\bar{\gamma}_{\pi(n)}}(x)$
respectively. It is clear that $F_{\bar{\gamma}_{\pi(n)}}(x) \leq
F_{{\gamma}_{\pi(n)}}(x) \leq F_{\tilde{\gamma}_{\pi(n)}}(x)$, where
the equalities hold when $n = 1$. Then, we have the following lemma:
\begin{lemma}\label{lemma:gamma_bound_ZFDPC}
The random variables $\tilde{\gamma}_{\pi(n)}$ and
$\bar{\gamma}_{\pi(n)}$, $n \in \{ 2,\cdots, M \}$, satisfy
\begin{eqnarray}\label{eq:extre_gamma_til}
&&\text{Pr}\{ u_n - \log\log \sqrt{K} \leq \tilde{\gamma}_{\pi(n)}
\leq  u_n + \log\log \sqrt{K}\} \nonumber\\
&& \hspace{3cm}\geq 1 - O\bigg(\frac{1}{\log K}\bigg),\\
\label{eq:extre_gamma_bar} &&\text{Pr}\{\chi_n - \log\log \sqrt{K}
\leq \bar{\gamma}_{\pi(n)} \leq \chi_n + \log\log \sqrt{K}\}\nonumber\\
&& \hspace{3cm} \geq 1 - O \left (\frac{1}{\log K}\right),
\end{eqnarray}
where\footnote{Here $\log(\cdot)$ represents the natural logarithm.}
\begin{align}\label{eq:u_n}
u_n = \log \bigg(\frac{K}{\varepsilon_n}\bigg) + (M+N-n-1) \log\log
\bigg(\frac{K }{\varepsilon_n}\bigg), \end{align} \begin{align}
\label{eq:chi_n} \chi_n = \log \left(\frac{K}{\epsilon_n}\right) +
(M+N-n-1) \log\log\left(\frac{K}{\epsilon_n}\right).
\end{align}
\end{lemma}
\begin{proof}
This result is readily established by combining
(\ref{eq:upper_tail_distri}) and (\ref{eq:lower_tail_distri}) with
the extreme order statistics result given in\footnote{Note that
there are some minor typographical errors with \cite[\emph{Lemma}
7]{Mohammad08}. Here we have adopted the correct results.}
\cite[\emph{Lemma} 7]{Mohammad08}.
\end{proof}
For the case $n = 1$, $\gamma_{\pi(n)}=
\tilde{\gamma}_{\pi(n)}=\bar{\gamma}_{\pi(n)}$, whose asymptotic
distribution is \cite{Mohammad08}
\begin{eqnarray}
&&\text{Pr}\{ u_1 - \log\log \sqrt{K} \leq \gamma_{\pi(1)} \leq u_1
+
\log\log \sqrt{K}\} \nonumber\\
&& \hspace{3cm}\geq 1 - O\bigg(\frac{1}{\log K}\bigg).
\end{eqnarray}
Interestingly, we can obtain the same result if we substitute $n=1$ into
(\ref{eq:extre_gamma_til})--(\ref{eq:chi_n}). The asymptotic
distribution of $\zeta_{\pi(n)}$ follows from the above results.
\begin{lemma}\label{lemma:SINR_bound_ZFDPC1}
Let $\rho = \frac{P}{M}$. For $\zeta_{\pi(n)}$, $n \in \{ 1,\cdots,M \}$, we
have
\begin{eqnarray}\label{eq:sinr_up_low2}
&&\hspace{-1cm}\text{Pr}\{ \varpi_n - \rho \log\log \sqrt K \leq
\zeta_{\pi(n)}
\leq \upsilon_n + \rho\log\log \sqrt K  \} \nonumber\\
&& \hspace{2cm}\geq 1 - O\bigg( \frac{1}{\log K} \bigg),
\end{eqnarray}
where
\begin{align}\label{eq:varpi_n}
\varpi_n= \rho \log \bigg(\frac{K}{\varepsilon_n}\bigg) + \rho
(M+N-n-1) \log\log \bigg(\frac{K }{\varepsilon_n}\bigg),
\end{align}
\begin{align}
\label{eq:upsilon_n} \upsilon_n = \rho \log
\left(\frac{K}{\epsilon_n}\right) + \rho (M+N-n-1)
\log\log\left(\frac{K}{\epsilon_n}\right).
\end{align}
\end{lemma}
\begin{proof}
See Appendix \ref{proof_SINR_bound_ZFDPC1}.
\end{proof}

We can now prove the following theorem (see
Appendix \ref{proof_theorem:sum_rate_GZFDPC}), which presents a key contribution:
\begin{theorem}\label{theorem:sum_rate_GZFDPC}
For a fixed number of transmit antennas $M$ and receive antennas
$N$, and fixed transmit power $P$, if the semi-orthogonality
parameter $\delta$ is chosen to satisfy (\ref{eq:Conditions}), then
the sum rate $R_{\text{ZFDPC-SUS}}$ of the proposed ZFDPC-SUS scheme
satisfies
\begin{eqnarray}\label{eq:converge_1}
\lim_{K \rightarrow \infty}
\frac{R_{\text{ZFDPC-SUS}}}{M\log_2[\rho\log K]} = 1
\end{eqnarray}
with probability 1, where $\rho = P/M$. In addition,
\begin{eqnarray}\label{eq:converge_2}
\lim_{K \rightarrow  \infty} \mathcal{E} \{R_{\text{BC}}\} -
\mathcal{E} \{ R_{\text{ZFDPC-SUS}}\} = 0,
\end{eqnarray}
where $R_{\text{BC}}$ denotes the sum rate of the MIMO broadcast
channel, achieved with DPC.  As $K\rightarrow \infty$, the average
sum rate difference between ZFDPC-SUS and DPC is no greater than
$O\big(\frac{\log \log K}{\log K} \big)$.
\end{theorem}

Note that the sum rate difference convergence
(\ref{eq:converge_2}) is much stronger than the sum rate ratio
convergence in probability (\ref{eq:converge_1}), since the latter
does not preclude the existence of an infinite sum rate gap between the
proposed scheme and the optimal scheme.

\subsection{ZFBF-SUS Scheme}\label{sec:linear_analysis}
In this section, we will evaluate the performance of linear ZFBF
with SUS. For our analysis, following \cite{Yoo06}, we will assume
that the criterion (\ref{eq:selection2}) is used at each iteration
of the SUS algorithm to select the best user. In \cite{Yoo06}, it
has been proved that ZFBF-SUS can achieve the same asymptotic sum
rate scaling as DPC. Here we establish the stronger result that the
average sum rate of ZFBF-SUS converges to the average sum rate
achieved with optimal DPC, which was not established in
\cite{Yoo06}. Deriving an exact expression for the asymptotic
distribution of the output SNR for each data stream, analogous to
(\ref{eq:sinr_up_low2}), appears very difficult for ZFBF-SUS. Thus,
here we adopt a different approach, based on first applying an upper
bound which relates the output SNR of ZFBF-SUS in terms of the
output SNR of ZFDPC-SUS, and then applying results from the previous
subsection.  This leads to the following key theorem:
\begin{theorem}\label{theorem:sum_rate_GZFBF}
For a fixed number of transmit antennas $M$ and receive antennas $N$, and fixed
transmit power $P$, if the semi-orthogonality parameter $\delta$ is chosen to satisfy (\ref{eq:Conditions}), then the sum rate $\mathcal{E} \{R_{\text{ZFBF-SUS}}\}$ of the ZFBF-SUS scheme satisfies:
\begin{align} \label{eq:SumRate_Diff}
\lim_{K \rightarrow \infty}\mathcal{E} \{R_{\text{BC}}\} -
\mathcal{E} \{R_{\text{ZFBF-SUS}}\} = 0 \; .
\end{align}
As $K\rightarrow \infty$, the average sum rate difference between ZFBF-SUS and DPC is no greater than
$O\big(\frac{\log \log K}{\log K} \big)$.
\end{theorem}
\begin{proof}
See Appendix \ref{proof_theorem:sum_rate_ZF}.
\end{proof}

This result shows that, as for the ZFDPC-SUS scheme, we can
significantly reduce the complexity of the SUS search algorithm by
choosing $\delta$ reasonably small, whilst at the same time achieve
the optimal asymptotic sum rate of DPC.

\subsection{Discussion of Results} \label{sec:Comparison}

Based on the analysis above, some interesting observations are readily in
order.
\begin{enumerate}
\item
Asymptotically, both schemes can achieve the maximum spatial
multiplexing gain of $M$, and also the maximum multi-user diversity
gain up to first order (i.e.\ the SNR scales with $\log K$, and the
sum rate scales as $\log \log K$). For ZFBF, this scaling behavior
agrees with previous results \cite{Alireza08,Love08}.

\item
As shown in \emph{Theorem} \ref{theorem:sum_rate_GZFDPC} and
\emph{Theorem} \ref{theorem:sum_rate_GZFBF}, provided that the
semi-orthogonality parameter $\delta$ is selected appropriately, the
asymptotic ergodic sum rates of both schemes converge to that of the
MIMO broadcast channel,
and in both cases the difference in average sum rate with respect to
optimal DPC is no greater than $O\left(\frac{\log \log K}{\log K}
\right)$. Note that similar scaling results have also been obtained
for other user selection schemes with ZFBF\cite{Alireza08,Love08}.

\item
In contrast to most related work, our results provide key insights
into the effect of the SUS semi-orthogonality parameter $\delta$ and
the number of receive antennas $N$. Considering ZFDPC-SUS, from
(\ref{eq:sinr_up_low2}) and the expressions for $\varpi_n$ in
(\ref{eq:varpi_n}) and $\upsilon_n$ in (\ref{eq:upsilon_n}), we see
that imposing the constraint $\delta$ does \emph{not reduce the
multi-user diversity gain in both first order terms $O(\log K)$ and
second-order terms $O(\log\log K)$}. It appears that this result can
not be established based on previous (less accurate) SUS analysis
methods \cite{Yoo06}. Moreover, our analysis demonstrates that
whilst the first order terms $O(\log K)$ in the multi-user diversity
gain are unaffected by the number of receive antennas $N$, the
second-order term grows linearly with both $N$ and $M$. This is
consistent with a similar conclusion made in \cite{Mohammad08},
which considered a different system configuration. 


\item
We can also draw insights into the design of $\delta$.
For practical systems with
\emph{finite} numbers of users, obtaining the exact $\delta$ which
yields the optimal complexity--performance tradeoff remains a
challenging open problem.
However, our asymptotic analysis still provides guidance for the
implementation of practical SUS algorithms. In particular, we see
that the choice of $\delta$ is closely related to $K$ and $M$ and,
to minimize complexity, it is clearly desirable to select $\delta$
to decrease with increasing $K$. At the same time, however, for
finite numbers of users it is advisable to ``overcompensate''  and
select $\delta$ to easily meet the conditions in
(\ref{eq:Conditions}). In our numerical experiments, we found that
for systems with $M \leq 8$, the choice of $\delta = \frac{1}{\log
K}$ can work well. In addition, since the number of candidate users
decreases with each iteration of the SUS algorithm, further
complexity savings can be achieved by adaptively selecting $\delta$;
e.g., at iteration $n$, setting $\delta_n =\frac{1}{\log
|\mathcal{U}_n|}$.

\item
Although the results in Section \ref{sec:perforance_GZFDPC} and \ref{sec:linear_analysis} demonstrate that both the ZFDPC-SUS and ZFBF-SUS schemes achieve the same asymptotic average sum rate, the speed of convergence to this optimal sum rate can be very different. Intuitively, this performance difference is caused by a reduction in the \emph{effective channel gain}\cite{Yoo06} seen by the ZFBF receivers. 
Thus, for finite $K$, there will be a gap in the average sum rates
of the two schemes. We will now study this more closely.
\end{enumerate}

\section{Sum Rate Analysis -- Finite $K$}\label{sec:finite_K_performance}
In this section, we analyze the achievable sum rates of the
ZFDPC-SUS and ZFBF-SUS schemes for \emph{finite} numbers of users.
To obtain clear insights, we focus on the high and low SNR regimes.
Our analysis is based on studying the gap between the sum rates
achieved by the two transceivers and a fixed upper bound.
This study follows the method of \cite{XZhang03}, which considered
single-user MIMO receivers. We will first evaluate the performance
for a given set of channel realizations, and then investigate the
average performance via simulations. We make the same assumptions as
stated at the beginning of Section \ref{sec:analysis}.

Given a set of $M$ users $\pi$ determined by user
selection\footnote{For a meaningful comparison, we will assume that
for both schemes, the same SUS selection criteria is used, based on
(\ref{eq:selection2}).  As such, the active users sets and the
corresponding compound channel matrix $\mathbf{C}_{\pi, d}$ will be
the same for both schemes.}, the sum capacity of the MIMO broadcast
channel $\{\mathbf{H}_{\pi(k)}\}_{k=1}^{M}$ can be written by using
the duality of the MIMO broadcast channel and the MIMO multiple
access channel as \cite{Vishwanath03} $C_{\text{BC}}
\big(\{\mathbf{H}_{\pi(k)}\}_{k=1}^{M}, P\big) =
\max_{\sum_{k}\text{tr}\mathbf{Q}_k \leq P} \log_2
\det\bigg(\mathbf{I} + \sum_{k=1}^{M} \mathbf{H}_{\pi(k)}^H
\mathbf{Q}_{k} \mathbf{H}_{\pi(k)}\bigg)$. Since no closed-form
solution exists, it is very difficult to compare $C_{\text{BC}}
\big(\{\mathbf{H}_{\pi(k)}\}_{k=1}^{M}, P\big)$ with
$R_{\text{ZFDPC-SUS}}$ and $R_{\text{ZFBF-SUS}}$. In fact, even with
our assumption of equal power allocation, i.e. $\mathbf{Q}_k =
\frac{P}{K N} \mathbf{I} $, this problem is still difficult, due to
the complicated structure of the compound channel matrix
$\mathbf{C}_{\pi, d}$ for the ZFDPC and ZFBF schemes (see
(\ref{eq:decoded_signal})). Thus, to analyze the difference in sum
rate between $R_{\text{ZFDPC-SUS}}$ and $R_{\text{ZFBF-SUS}}$ for
finite $K$, we adopt an indirect approach and focus on
characterizing the differences between the sum
rates achieved by the two transceiver structures and $C$, 
where
$C = \log_2 \det(\mathbf{I}_M + \rho\mathbf{C}_{\pi, d}\mathbf{C}_{\pi,d}^H)$ with $\rho = P/M$.

Before presenting our main results, it is worth noting that \cite[\emph{Theorem} 3]{Caire03} $\lim_{P\rightarrow \infty} C_{\text{BC}}(\mathbf{C}_{\pi, d}, P)- C
= 0$, where $C_{\text{BC}}(\mathbf{C}_{\pi, d}, P)$ denotes the sum
capacity of a MIMO broadcast system given by
(\ref{eq:decoded_signal}). Moreover, for the case  $N=1$, $\{\mathbf{H}_{\pi(k)}\}_{k=1}^{M}$
reduces to $\mathbf{C}_{\pi, d}$ and $C_{\text{BC}}
\big(\{\mathbf{H}_{\pi(k)}\}_{k=1}^{M}, P\big)$ coincides with
$C_{\text{BC}}(\mathbf{C}_{\pi,d }, P)$.
Thus, the high SNR results which we establish below correspond precisely to the gaps between the sum rates achieved by the two
transceivers and the sum capacity achieved with optimal DPC.
Define
\begin{eqnarray} \label{eq:etaKappa}
\eta_i = \sum_{j=1}^{i-1} \frac{|l_{i,j}|^2}{|l_{i,i}|^2},~~
\kappa_i = \sum_{j= i+1}^{M} \frac{|t_{j,i}|^2}{|t_{i,i}|^2},
\end{eqnarray}
where $l_{i,j}$ and $t_{i,j}$ are the $(i,j)$-th elements of
matrices $\mathbf{L}_{\pi, d}$ and $\mathbf{T}_{\pi,d }$,
respectively. Some basic manipulations of the results in
\cite{XZhang03} yield the following theorem:
\begin{theorem}\label{theorem:gaps_finite_K}
For finite number of users $K$, finite number of transmit and
receive antennas $M$ and $N$,
\begin{itemize}
\item
In the high SNR region:
\begin{eqnarray}
\label{eq:h_gap_ZFDPC_finite} C - R_{\text{ZFDPC-SUS}} &=&
\frac{1}{\rho \log 2}\sum_{i=1}^{M}
\frac{\kappa_i}{\lambda_{\pi(i)}|l_{i,i}|^2} \nonumber\\
&& + O(\rho^{-2}),\\
\label{eq:h_gap_ZFBF_finite} C - R_{\text{ZFBF-SUS}} &=&
\sum_{i=1}^{M} \log_2(1+ \kappa_i)\nonumber\\ &&+ O(\rho^{-2}).
\end{eqnarray}
\item
In the low SNR region:
\begin{eqnarray}
\label{eq:gap_ZFDPC_finite} C - R_{\text{ZFDPC-SUS}} &=&
\frac{\rho}{\log 2}\sum_{i=1}^{M} \eta_i \lambda_{\pi(i)}|l_{i,i}|^2
\nonumber\\ && +O(\rho^2),
\end{eqnarray}
\begin{eqnarray} \label{eq:gap_ZFBF_finite}
C - R_{\text{ZFBF-SUS}} &=& \frac{\rho}{\log 2}\sum_{i=1}^{M} (1+
\eta_i - \frac{1}{1+\kappa_i}) \nonumber\\
&& \times ~ \lambda_{\pi(i)}|l_{i,i}|^2 + O(\rho^2).
\end{eqnarray}
\end{itemize}
\end{theorem}

From these results, we can make the following conclusions.

\emph{High SNR Region:} As $\rho \to \infty$, for ZFDPC-SUS the sum
rate approaches $C$, whereas for ZFBF-SUS there is a constant sum
rate gap of $\mathcal{A} \triangleq \sum_{i=1}^{M} \log_2(1+
\kappa_i)$. This gap can be zero only when $\kappa_i = 0$, which is
a rare case corresponding to complete orthogonality between the row
vectors of $\mathbf{C}_{\pi, d}$.  Subtracting
(\ref{eq:gap_ZFDPC_finite}) from (\ref{eq:gap_ZFBF_finite}), in this
region we can also quantify the sum rate gap between ZFDPC-SUS and
ZFBF-SUS as $R_{\text{ZFDPC-SUS}} - R_{\text{ZFBF-SUS}} =
\mathcal{A} + O(\rho^{-1})$, which shows the advantage of ZFDPC-SUS
for finite $K$.

\emph{Low SNR Region:} As $\rho \to 0$, for both ZFDPC-SUS and
ZFBF-SUS, the sum rate gaps w.r.t.\ $C$ approach zero linearly with
$\rho$. Moreover, in this region we can again quantify the sum rate
gap as $R_{\text{ZFDPC-SUS}} - R_{\text{ZFBF-SUS}} =
\frac{\rho}{\log 2}\sum_{i=1}^{M} (1 - \frac{1}{1+\kappa_i})
\lambda_{\pi(i)}|l_{i,i}|^2$, which is non-negative.  It is also
worth noting that in the low SNR regime, better performance may be
achievable by transmitting with full power to only a single user,
rather than sending equal power streams to $M$ selected users. The
benefit of this approach, however, will depend not only on the SNR
value, but also on the number of users $K$.  In particular, the
benefit of this approach is expected to be most evident when $K$ is
small, for which case there will be the most disparity between the
dominant eigen-channels of the users.


\emph{Effect of SUS Parameter $\delta$:} According to the SUS algorithm, we have $|l_{i,j}|^2 < \delta$ for $i > j$, and $|l_{i,i}|^2 > 1-(i-1)\delta$.  Thus, with smaller semi-orthogonality parameter $\delta$, it is more
likely to have off-diagonal elements with smaller absolute value in
both $\mathbf{L}_{\pi,d}$ and $\mathbf{T}_{\pi,d}$ (i.e smaller
$|l_{i,j}|, i<j$ and $|t_{j,i}|, i < j$ ) and more likely to have
diagonal elements with larger absolute value in $\mathbf{L}_{\pi,d}$. From (\ref{eq:etaKappa}), these observations imply that a smaller $\delta$ leads to smaller $\eta_i$ and $\kappa_i$. In addition, 
it is easy to see that
$\eta_i |l_{i,i}|^2 = \sum_{j=1}^{i-1}|l_{i,j}|^2$ and $(1+ \eta_i) |l_{i,i}|^2 = 1$.
With these results, we see that by decreasing $\delta$, the sum rate gaps for both transceivers are
likely to decrease, for both high and low SNRs. This implies that the sum rates of
both transceivers are likely to increase, which agrees with intuition.

\begin{figure}[t]
\centering
\includegraphics[width= 0.9\columnwidth]{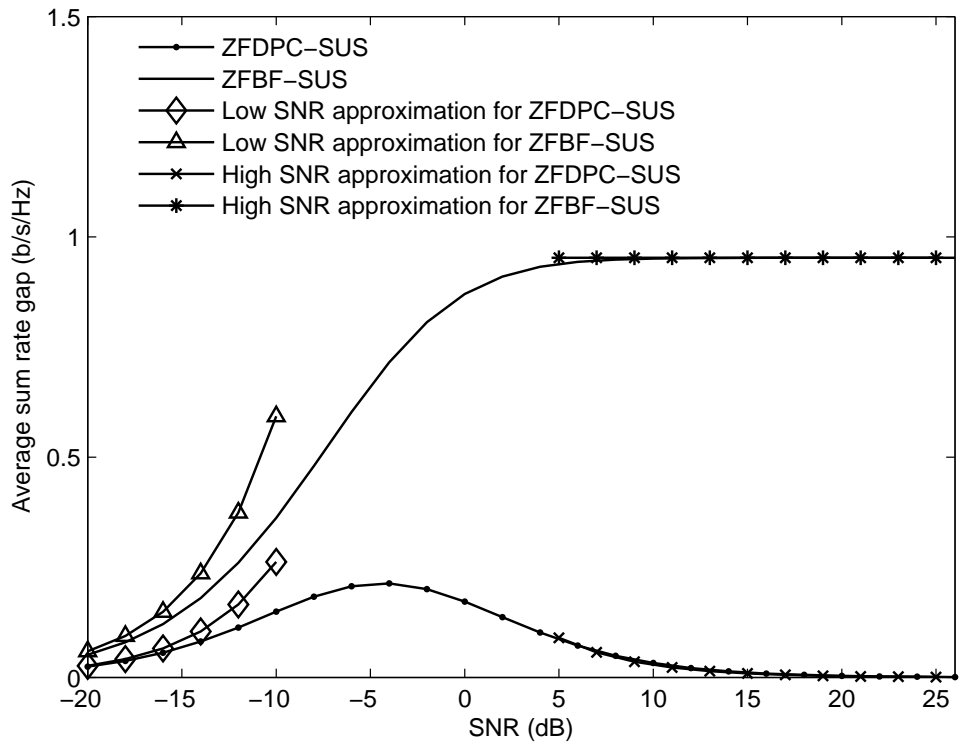}
\caption{Comparison of sum rate gap for different SNRs. $M=4$,
$N=4$, $K = 50$.} \label{fig:sum_rate_gap}
\end{figure}

Fig. \ref{fig:sum_rate_gap} demonstrates the average sum rate gaps
of ZFDPC-SUS and ZFBF-SUS for different SNRs. Results are shown for
$M = 4$, $N = 4$, $K=50$, and $\delta = \frac{1}{\log K}$. These
results confirm our analytical conclusions given above, based on
\emph{Theorem} \ref{theorem:gaps_finite_K}.

\begin{figure}[t] \centering
\includegraphics[width=0.9\columnwidth]{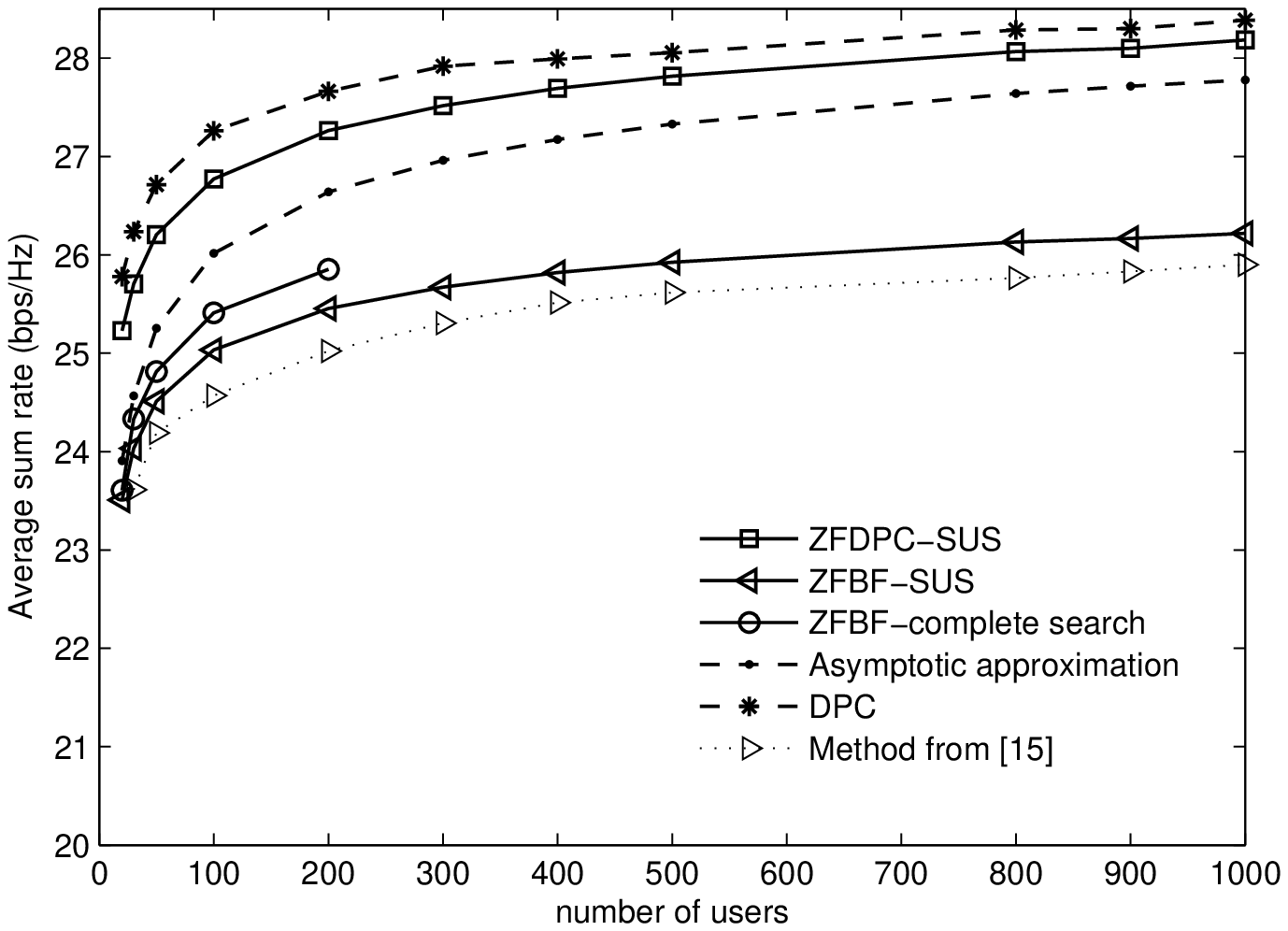}
\caption{Comparison of average sum rates for different numbers of
users. $M=4$, $N=4$, $P = 15$ dB.} \label{fig:sum_rate}
\end{figure}

\begin{figure}[t]
\centering
\includegraphics[width= 0.9\columnwidth]{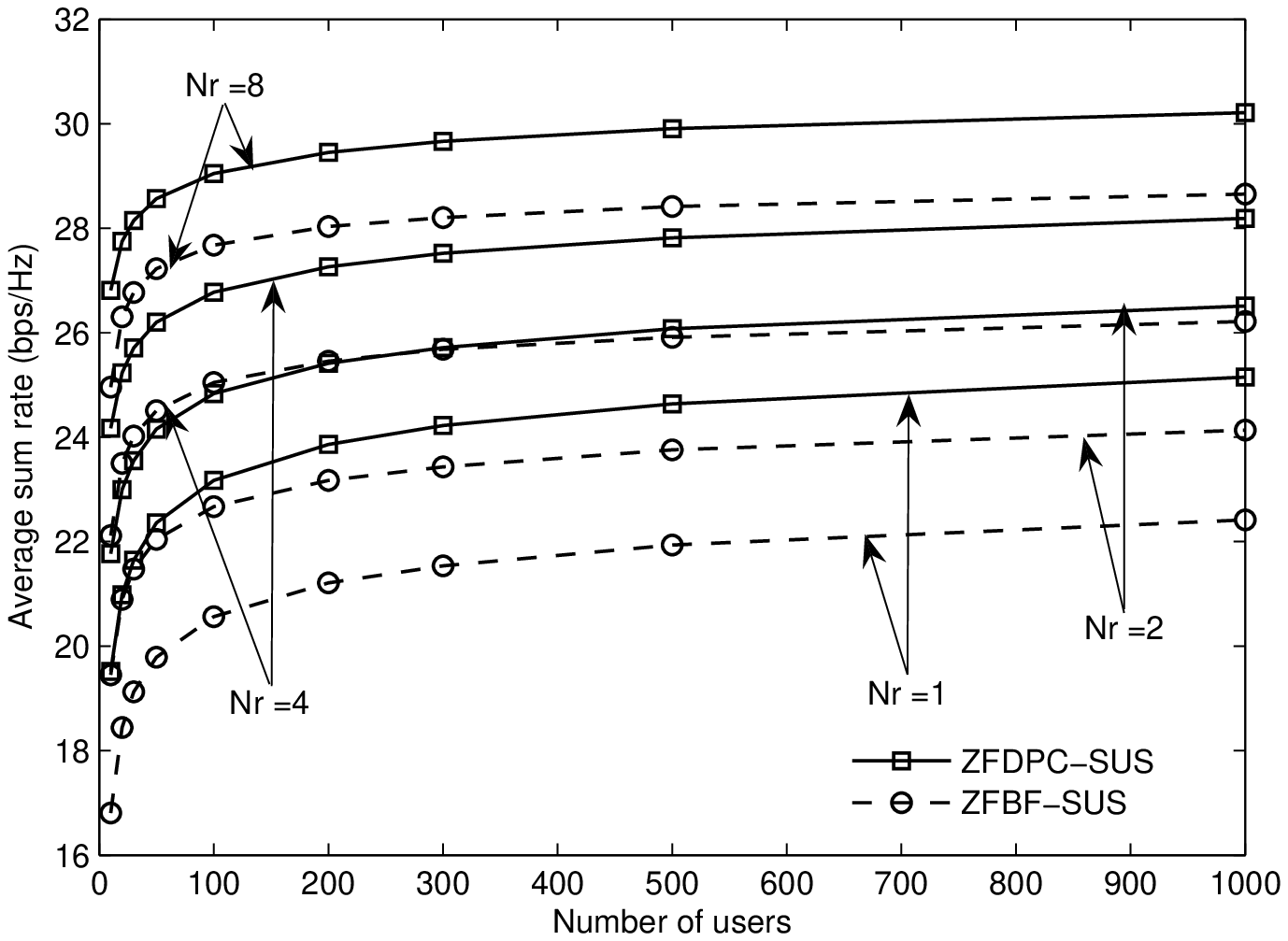}
\caption{Comparison of average sum rates for different numbers of
users and different numbers of receive antennas. $M=4$, $P = 15$
dB.} \label{fig:sum_rate_Nr}
\end{figure}

\section{Numerical Results}\label{sec:numerical_result}

For our simulations, we use $P = 15$ dB, $\delta = \frac{1}{\log K}$, and the optimal water-filling power allocation.


Fig. \ref{fig:sum_rate} plots the average sum rate achieved by
ZFDPC-SUS and ZFBF-SUS as a function of the number of users. Curves
are also presented for ZFBF with complete search, as well as optimal
DPC.  In the first case, a search is conducted over all combinations
of users, and the combination with the highest sum rate is selected.
Due to the very high complexity of this approach, we only provide
results for relatively small $K$. The optimal DPC curve acts as an
achievable upper bound, and is computed using the algorithm from
\cite{Jindal05}.  In addition, based on (\ref{eq:sinr_up_low}) and
the expressions for $u_n$ in (\ref{eq:u_n}) and $\chi_n$ in
(\ref{eq:chi_n}), we have plotted $\sum_{i=1}^{M}\log_2(1+ \rho
(\log K + (M+N-i-1)\log\log K))$ as an asymptotic approximation for
the average sum rate of the ZFDPC-SUS scheme. As evident from the
figure, the performance of ZFDPC-SUS is very close to that of DPC,
and is slowly converging to DPC as $K$ grows large. The asymptotic
approximation for ZFDPC-SUS based on our analysis is also quite good
(within $1$ bps/Hz). Considering ZFBF, we see that the ZFBF-SUS
curve is no more than $0.5$ dB away from that of the complete search
method; further verifying the utility of the SUS approach. Moreover,
the ZFBF curves are far below the ZFDPC-SUS curve, demonstrating
that ZFDPC-SUS has \emph{significant} performance advantages at
finite $K$. For further comparison, we have also implemented a
related algorithm proposed in \cite{Alireza08} and plotted the
corresponding sum rate curve. This curve is generated by
using an optimal threshold, 
computed by an
exhaustive search. The performance is close to that of ZFBF-SUS.

Fig. \ref{fig:sum_rate_Nr} compares the average sum rate of
ZFDPC-SUS and ZFBF-SUS as a function of the number of users, for
different numbers of receive antennas. Note that according to
(\ref{eq:sinr_up_low}) and the expressions for $u_n$ and $\chi_n$ in
(\ref{eq:u_n}) and  (\ref{eq:chi_n}) respectively, if we increase
the number of receive antennas by one, the increase in sum rate can
be approximated as $M \log \bigg(1+ \frac{\rho \log \log K}{1+\rho
\log K}\bigg) \rightarrow 0 $ as $K \rightarrow \infty$; i.e., the
difference in sum rate will be negligible for large $K$. However,
the figure shows that this convergence is very slow, and that
increasing the number of receive antennas can significantly increase
the sum rate for finite $K$.

\section{Conclusion}

We have investigated the sum rate of two low complexity
eigenmode-based transmission techniques for the MIMO broadcast
channel, ZFDPC-SUS and ZFBF-SUS. We proved that ZFDPC-SUS can
achieve the optimal sum rate scaling of the MIMO broadcast channel,
and that the average sum rate of both techniques converges to the
average sum capacity of the MIMO broadcast channel as $K$ grows
large (albeit at different rates). We also investigated and compared
the achievable sum rates of ZFDPC-SUS and ZFBF-SUS for finite $K$,
and demonstrated that ZFDPC-SUS has significant performance
advantages. In contrast to most previous related results, our
analytical results provide important insights into the benefit of
multiple receive antennas, and the effect of the SUS algorithm.

\appendices

\section{Proof of \emph{Lemma} \ref{lemma:iid_property}}\label{proof_lemma:iid_property}

Our derivation closely follows the method of proof for \cite[\emph{Lemma} 3]{Love08}
and \cite[\emph{Lemma} 1]{Wang_report07}. For two complex vectors $\mathbf{z} =
\mathbf{z}_r + \jmath  \mathbf{z}_i$ and $\mathbf{z}^{\prime} =
\mathbf{z}^{\prime}_r + \jmath  \mathbf{z}^{\prime}_i$ with the same
dimension, we write $\mathbf{z} \preceq \mathbf{z}^{\prime}$ if
every element of $\mathbf{z}_r $ and $\mathbf{z}_i $ is less than or
equal to its counterpart in $\mathbf{z}^{\prime}_r $ and
$\mathbf{z}^{\prime}_i $, respectively. Let $\mathcal{K}_{n}$ denote
the cardinality of the candidate set $\mathcal{U}_n$. For the first
iteration, $\mathcal{K}_{1} = K$ and $\mathbf{c}_{\pi(1)}$ is the
vector with the maximum norm. For clarity of exposition, at the end of $n$-th iteration, we
relabel the eigen-channel vectors in $\mathcal{U}_n/\{\pi(n)\}$ as
$\mathbf{\tilde{c}}_1,\cdots, \mathbf{\tilde{c}}_{\mathcal{K}_n-1}$.

We find that the result in \cite[\emph{Lemma} 1]{Wang_report07},
which was derived specifically for Gaussian vectors, holds more
generally and does not require the Gaussian assumption, and indeed
can also be adapted to our case. The proof is based on induction.
For the first iteration, we have
\begin{eqnarray}
&&\hspace{-1cm}\text{Pr}\{\mathbf{\tilde{c}}_1 \preceq \mathbf{z}_1,
\cdots, \mathbf{\tilde{c}}_{K-1} \preceq \mathbf{z}_{K-1}|
\mathbf{c}_{\pi
(1)} = \mathbf{z}_{(1)}\} \nonumber\\
&=& \prod_{i=1}^{K-1} \text{Pr} \{\mathbf{\tilde{c}}_i  \preceq
\mathbf{z}_i| \|\mathbf{\tilde{c}}_i\|< \|\mathbf{z}_{(1)}\| \}
\end{eqnarray}
and since $\lim_{K \rightarrow \infty} \|\mathbf{z}_{(1)}\| =
\infty $,
\begin{eqnarray}
\lim_{K \rightarrow \infty}\text{Pr} \{\mathbf{\tilde{c}}_i \preceq
\mathbf{z}_i |\|\mathbf{\tilde{c}}_i\|< \|\mathbf{z}_{(1)}\| \} =
F_{\mathbf{c}} (\mathbf{z}_i),
\end{eqnarray}
where $F_{\mathbf{c}} (\cdot)$ is the c.d.f.\ of the principal
eigen-vector of a complex Wishart matrix.

Now assume that this lemma holds up to the $(n-1)$-th iteration and
let us consider the $n$-th iteration. Conditioned on
$\mathbf{c}_{\pi(1)}, \cdots, \mathbf{c}_{\pi(n-1)}$, according to
our assumption, the channel vectors in $\mathcal{U}_n$ are i.i.d.
and converge in distribution to the principal eigen-vector of a
complex Wishart matrix. At the end of step 3) of the $n$-th
iteration, user $\pi(n)$ is chosen. Any user $k$ in $\mathcal{U}_n$
satisfies $ \gamma_{k} (n) \leq \gamma_{\pi(n)}$. Replacing the
condition\footnote{To be more precise, we note that different
notation is used in \cite{Love08}. Our conditions
$\{\mathbf{c}_{\pi(1)}= \mathbf{z}_{(1)}, \cdots,
\mathbf{c}_{\pi(n)} = \mathbf{z}_{(n)}\}$ and $\{\mathbf{c}_{\pi(1)}
= \mathbf{z}_{(1)}, \cdots, \mathbf{c}_{\pi(n-1)} =
\mathbf{z}_{(n-1)}, \gamma_{k}(n) \leq \gamma_{\pi(n)} \}$ are
analogous to the conditions $\{\mathbf{h}_{j_{(1)}}=
\mathbf{z}_{(1)}, \cdots, \mathbf{h}_{j_{(n)}} = \mathbf{z}_{(n)}\}$
and $\{\mathbf{h}_{j_{(1)}} = \mathbf{z}_{(1)}, \cdots,
\mathbf{h}_{j_{(n-1)}} = \mathbf{z}_{(n-1)}, R_{(n)}^{\text{BF}}
\left( \mathbf{h}_{i} \right) \leq R_{(n)}^{\text{BF}} \left(
\mathbf{z}_{(n)} \right) \}$ given in \cite{Love08}.}
$\{\mathbf{c}_{\pi(1)} = \mathbf{z}_{(1)} \}$ and $\{ \|
\mathbf{\tilde{c}}_i\| \leq  \| \mathbf{z}_{(1)} \| \}$ by $\{
\mathbf{c}_{\pi(1)} = \mathbf{z}_{(1)}, \cdots,
\mathbf{c}_{\pi(n-1)} = \mathbf{z}_{(n-1)}, \mathbf{c}_{\pi(n)} =
\mathbf{z}_{(n)}\}$ and $\{ \mathbf{c}_{\pi(1)} = \mathbf{z}_{(1)},
\cdots, \mathbf{c}_{\pi(n-1)} = \mathbf{z}_{(n-1)},  \gamma_{k}(n)
\leq \gamma_{\pi(n)} \}$ respectively in the derivation in
\cite[\emph{Lemma} 1]{Wang_report07} and following the same method
as in \cite[\emph{Lemma} 1]{Wang_report07}, we can establish that
the remaining channel vectors in $\mathbf{\mathcal{U}}_n$ are
i.i.d.\ with c.d.f.
\begin{eqnarray}
&&\text{Pr} \{\mathbf{\tilde{c}}_i \preceq \mathbf{z}_i |
\mathbf{c}_{\pi(1)} = \mathbf{z}_{(1)},  \cdots,\nonumber \\
&& \hspace{1cm}\mathbf{c}_{\pi(n-1)} = \mathbf{z}_{(n-1)},
\gamma_{k} (n) \leq \gamma_{\pi(n)} \}
\end{eqnarray}
for $i = 1, \ldots, {\mathcal{K}_n-1}$. Since $\lim_{K \rightarrow \infty} \mathcal{K}_n = \infty $,
$\gamma_{\pi(n)}$ is unbounded from above, i.e.,
\begin{eqnarray}
\lim_{K \rightarrow \infty}\gamma_{\pi(n)} = \infty,
\end{eqnarray}
and we have
\begin{eqnarray}
&&\hspace{-0.5cm}\lim_{K \rightarrow \infty} \text{Pr}~
\{\mathbf{\tilde{c}}_i \preceq \mathbf{z}_i | \mathbf{c}_{\pi(1)} =
\mathbf{z}_{(1)},
\cdots,\nonumber \\
&& \hspace{1cm} \mathbf{c}_{\pi(n-1)} = \mathbf{z}_{(n-1)},
\gamma_{k}(n) \leq \gamma_{\pi(n)} \} \nonumber \\ && = \text{Pr}
~\{\mathbf{\tilde{c}}_i \leq \mathbf{z}_i | \mathbf{c}_{\pi(1)} =
\mathbf{z}_{(1)}, \cdots, \mathbf{c}_{\pi(n-1)} =
\mathbf{z}_{(n-1)}\}.\nonumber\\
\end{eqnarray}
By induction $ \text{Pr} ~\{\mathbf{\tilde{c}}_i \preceq\mathbf{z}_i
| \mathbf{c}_{\pi(1)} = \mathbf{z}_{(1)}, \cdots,
\mathbf{c}_{\pi(n-1)} = \mathbf{z}_{(n-1)}\} $ converges in
distribution to the distribution of the principal eigen-vector of a
complex Wishart matrix, thereby establishing the lemma.

\section{Proof of \emph{Lemma} \ref{lemma:distri_mu_delta}}\label{proof_lemma:distri_mu_delta}
According to \emph{Lemma} \ref{lemma:iid_property}, the eigen-vector
$\mathbf{v}_k$, for $k\in \mathcal{U}_n$, is an isotropically
distributed unit vector on the $M$-dimensional complex unit
hypersphere. In addition, for large $K$, the subspace spanned by the
orthonormal basis $\mathbf{q}_1, \cdots, \mathbf{q}_{n-1}$ becomes
independent of $\mathbf{v}_k$. Thus, without loss of generality we
can assume $\mathbf{q}_i= \mathbf{e}_i$, where $\mathbf{e}_i$ is the
$i$-th row of the identity matrix $\mathbf{I}_{M}$. Let
$\mathbf{v}_k = [v_1, \cdots, v_M]^T$, then
\begin{eqnarray} \label{eq:ProbDesired}
\mu_n(\delta) &=& \text{Pr}\left( |\mathbf{v}_k^H \mathbf{q}_1^H
|^2< \delta,\cdots, |\mathbf{v}_k^H \mathbf{q}_{n-1}^H |^2 < \delta
\right)\nonumber\\ &=& \text{Pr}\left(|v_1|^2< \delta,\cdots,
|v_{n-1}|^2< \delta \right).
\end{eqnarray}
In the following we will first derive the joint p.d.f. of $|v_1|^2,
\cdots, |v_{n-1}|^2$.

The surface area of a complex unit hypersphere of $M$ dimensions is
$\frac{2\pi^M}{\Gamma(M)}$ \cite{Kendall1961}. So the joint p.d.f.
of $v_1,\cdots,v_M$ can be written as:
\begin{equation}
f(\mathbf{v}_k) = f(v_1, \cdots, v_M) = \left\{ \begin{array}{ll}
\frac{\Gamma(M)}{2\pi^M} , & \|\mathbf{v}_k\|=1 \\
0, &\text{otherwise}
\end{array}
\right..
\end{equation}
Define $v_i= x_{2i-1}+ \jmath x_{2i}$. Then, the joint p.d.f.\ of $x_1,\cdots,
x_{2M}$ can be expressed as:
\begin{eqnarray}\label{eq:joint_pdf_x}
f(x_1,x_2,\cdots, x_{2M})= \left\{ \begin{array}{ll}
\frac{\Gamma(M)}{2\pi^M} , & \sum_{i=1}^{2M}x_i^2 = 1 \\
0, &\text{otherwise}
\end{array}
\right..
\end{eqnarray}
We require the joint p.d.f.\ of $x_1, \cdots, x_{2(n-1)}$, which is evaluated via
\begin{eqnarray} \label{eq:jointPDF_Int}
&& \hspace{-1cm} f(x_1, \cdots, x_{2(n-1)} ) \nonumber\\
& = & {\int \cdots \int}_{\sum_{i=1}^{2M}{x_i}^2=1} f(x_1, \cdots,
x_{2M})\nonumber\\
&& \times ~{\rm d} x_{2(n-1)+1} \cdots {\rm d} x_{2M} \;
\nonumber\\
&=&  \frac{\Gamma(M)}{2\pi^M} V(x_1, \cdots, x_{2(n-1)})
\end{eqnarray}
where $V( x_1, \cdots, x_{2(n-1)} )$ denotes the area
\begin{eqnarray} \label{eq:Volume}
&& \hspace{-1cm}V( x_1, \cdots, x_{2(n-1)} ) \nonumber\\
&=& {\int\cdots
\int}_{\sum_{i=1}^{2M}{x_i}^2=1} ~{\rm d} x_{2(n-1)+1} \cdots {\rm d} x_{2M} \; \nonumber \\
&=& {\int\cdots \int}_{\sum_{i=2(n-1)+1}^{2M}x_i^2 = 1-
\sum_{i=1}^{2(n-1)}{x_i}^2}\nonumber\\
&& \times ~{\rm d} x_{2(n-1)+1} \cdots {\rm d} x_{2M} \;.
\end{eqnarray}
The multi-dimensional integral (\ref{eq:Volume}) is seen to be the
surface area of a real $(2M-2(n-1))$-dimensional hypersphere of
radius $\sqrt{1- \sum_{i=1}^{2(n-1)}{x_i}^2}$.  Thus, using results
from \cite{Kendall1961}, we evaluate this integral as follows:
\begin{eqnarray}\label{eq:int_vol}
&& \hspace{-1cm} V( x_1, \cdots, x_{2(n-1)} ) \nonumber\\
&=& \frac {2\pi^{M-n+1}}{\Gamma(M-n+1)} \left(1 -
\sum_{i=1}^{2(n-1)}{x_i}^2\right)^{\frac{2(M-n+1)-1}{2}} \nonumber\\
&& \times \sqrt{ {\det}~ \mathbf{A} } ~{\rm d} x_1 \cdots {\rm d}
x_{2(n-1)} ,
\end{eqnarray}
where $\mathbf{A}$ is a $ (2(n-1)+1) \times (2(n-1)+1)$ matrix with
$(i,j)$-th element $\mathbf{A}_{i,j}= \frac{\partial \boldsymbol{\theta}}{\partial x_i}
\cdot \frac{\partial \boldsymbol{\theta}}{\partial x_j}$
with $\boldsymbol{\theta} = \bigg(x_1,\cdots, x_{2(n-1)}, \sqrt{1-
\sum_{i=1}^{2(n-1)} x_i^2}  \bigg)^T$, and `$\cdot$' denotes the vector inner
product operation. We can compute $\mathbf{A}_{i,j} = \delta_{i,j} +
\frac{ x_i x_j} {1 - \sum_{i=1}^{2m}{x_i}^2}
$, where $ \delta_{i,j}$ is the Kronecker-delta function, and after some
manipulations obtain $\det{\mathbf{A}} = \frac{1}{1 - \sum_{i=1}^{2(n-1)}{x_i}^2}$.
Combining this result with (\ref{eq:jointPDF_Int}) and (\ref{eq:int_vol}) we obtain
\begin{eqnarray}
f(x_1, \cdots, x_{2(n-1)} )
&=& \frac{\Gamma(M)}{\Gamma(M-n+1)\pi^{n-1}} \nonumber\\
&& \hspace{-2cm} \times \left( 1 - \sum_{i=1}^{2(n-1)}{x_i}^2
\right)^{M-n}.
\end{eqnarray}
It is now convenient to make the polar coordinate transformations $x_{2i-1} = r_i \cos\theta_i$, $x_{2i}=r_i \sin{\theta_i}
$, for $i=1, \cdots, n-1$, where $r_i \geq 0$, $0 \leq\theta_i \leq 2 \pi$.
The corresponding Jacobian is easily evaluated as \cite{Kendall1961} $\left( \prod_{i=1}^{n-1}
r_i \right)^{-1}$.  So the joint density of $r_1,\cdots,r_{n-1}$ is
\begin{eqnarray}
&& \hspace{-1cm}  f(r_1,\cdots, r_{n-1}) \nonumber\\ &=&
\frac{\Gamma(M)}{\Gamma(M-n+1) \pi^{n-1}} \left( 1 -
\sum_{i=1}^{n-1} r_i^2 \right)^{M-n} \prod_{i=1}^{n-1} r_i
\nonumber\\ && \times ~ \prod_{i=1}^{n-1} \int_{0}^{2\pi} {\rm d}
\theta_i \nonumber\\
&=& \frac{2^{n-1}\Gamma(M)}{\Gamma(M-n+1)} \left(1 - \sum_{i=1}^{n-1}
r_i^2 \right)^{M-n} \prod_{i=1}^{n-1} r_i .
\end{eqnarray}
Next we apply the transformation $t_i = r_i^2$, $i = 1, \ldots,
n-1$.  Clearly $t_i = |v_i|^2$ (we will deal with $t_i$ subsequently
to simplify notation). The corresponding Jacobian is $J
(t_1,\ldots,t_{n-1}) = 1/(2^{n-1} \sqrt{t_1,\cdots, t_{n-1}})$. So
we obtain the desired joint p.d.f.\ of $t_1,\ldots,t_{n-1}$ as
\begin{align}\label{eq:joint_density}
f (t_1,\ldots,t_{n-1}) = \frac{\Gamma(M)}{\Gamma(M-n+1)} \left(1 -
\sum_{i=1}^{n-1}t_i \right)^{M-n}.
\end{align}

Armed with this result, we can now evaluate the desired probability
$\mu_n(\delta)$ in (\ref{eq:ProbDesired}). For notational
convenience, we will consider $\mu_{n+1}(\delta)$, for $n+1 \in \{
2,\cdots, M \}$. Denoting  $D_{n}=\{ 0 \leq t_1 \leq \delta,\cdots,
0 \leq t_n \leq \delta \}$, we have
\begin{eqnarray} \label{eq:prob_sel_n}
\mu_{n+1} (\delta) &=&  \int \cdots \int_{D_{n}} f \left(t_1,
\cdots,t_n \right)~{\rm d}t_1 \cdots{\rm d}t_n \nonumber \\&=&
\frac{\Gamma(M)}{\Gamma(M-n)} \varphi_{n}(1)
%
\end{eqnarray}
where we have defined
\begin{equation}
\varphi_{n}(z) = \int \cdots \int_{D_{n}} \left(z -
\sum_{i=1}^{n}{t_i} \right)^{M-n-1} ~{\rm d}t_1 \cdots{\rm d}t_n
\end{equation}
for $z \geq n \delta$. Note that with this definition,
$\varphi_{n}(1)$ exists for all $n$ provided that $\delta <
\frac{1}{M-1}$. This condition is assumed in the lemma statement.
Then $\varphi_{n}(z)$ can be written as
\begin{eqnarray}\label{eq:varphi_n}
\varphi_{n}(z) &=& \int \cdots \int_{D_{n-1}}\left(\int_{0}^{\delta}
\left(z - \sum_{i=1}^{n}{t_i}\right)^{M-n-1} ~{\rm d}t_n \right)
~{\rm d}t_1 \cdots{\rm
d}t_{n-1} \nonumber\\
&=& \frac{1}{M-n}\int \cdots \int_{D_{n-1}} \left[ \left(z -
\sum_{i=1}^{n-1}{t_i}\right)^{M-n} - \left(z - \delta -
\sum_{i=1}^{n-1}{t_i}\right)^{M-n} \right] ~{\rm d}t_1 \cdots{\rm
d}t_{n-1}\nonumber\\
&=& \frac{1}{M-n} \left( \varphi_{n-1}(z) - \varphi_{n-1}(z- \delta)  \right) .
\label{eq:iter_relation}
\end{eqnarray}
So we have
\begin{align}\label{eq:induction_assumption_k2_1}
\varphi_{n}(1) &= \frac{1}{M-n} \big( \varphi_{n-1}(1) -
\varphi_{n-1}(1-
\delta)~ \big)\\
&= \frac{1}{(M-n)(M-n+1)}\nonumber\\
& \times (\varphi_{n-2}(1) - 2 \varphi_{n-2}(1 -\delta)
+\varphi_{n-2}(1 -2\delta) ). \label{eq:induction_assumption_k2}
\end{align}

We will now prove, using mathematical induction, that
for any integer $k \in \{1, 2, \cdots, n-1\}$,
\begin{eqnarray}\label{eq:induction_assumption1}
\varphi_{n}(1) &=& \bigg[ \prod_{j=0}^{k-1}(M-n+j)\bigg]^{-1}
\nonumber\\ && \times\sum_{i=0}^{k} (-1)^i \binom{k}{i}
\varphi_{n-k}(1-i \delta) .
\end{eqnarray}
According to (\ref{eq:induction_assumption_k2_1}) and
(\ref{eq:induction_assumption_k2}), (\ref{eq:induction_assumption1})
holds for $k=1$ and $k=2$ respectively. Assuming that
(\ref{eq:induction_assumption1}) holds for integer $k$, applying
(\ref{eq:iter_relation}) in (\ref{eq:induction_assumption1}) yields
\begin{eqnarray}\label{eq:HeadEq}
\varphi_{n}(1) & =& \bigg[ \prod_{j=0}^{k}(M-n+j)\bigg]^{-1}
\sum_{i=0}^{k} (-1)^i \binom{k}{i} \bigg[ \varphi_{n-k-1}(1 -
i~\delta) - \varphi_{n-k-1}(1-(i+1)~\delta)
\bigg]\\
&=& \bigg[ \prod_{j=0}^{k}(M-n+j)\bigg]^{-1} \bigg\{
\varphi_{n-k-1}(1) +
(-1)^{k+1} \varphi_{n-k-1}(1 - (k+1)~\delta)\nonumber\\
\label{eq:MidEq} &&\vspace{1cm} + \sum_{i=0}^{k-1} (-1)^{i+1}
\binom{k+1}{i+1}
\varphi_{n-k-1}(1 - (i+1)~\delta) \bigg\} \\
&=& \bigg[ \prod_{j=0}^{k}(M-n+j)\bigg]^{-1} \sum_{i=0}^{k+1}(-1)^i
\binom{k+1}{i} \varphi_{n-k-1}(1 - i~\delta)
\label{eq:EndEq}
\end{eqnarray}
where, to obtain (\ref{eq:MidEq}), we have used $\binom{k}{i+1} =
\binom{k-1}{i} +\binom{k-1}{i+1}$. Thus, from (\ref{eq:EndEq}), if
(\ref{eq:induction_assumption1}) holds for integer $k$, it also
holds for $k+1$. By induction, (\ref{eq:induction_assumption1}) then
holds for any integer $1 \leq k < n$. Setting $k = n-1$ in
(\ref{eq:induction_assumption1}),
\begin{eqnarray} \label{eq:varPhi_Induct}
\varphi_{n}(1) & =& \bigg[ \prod_{j=0}^{n-2}(M-n+j)\bigg]^{-1}
\nonumber\\ && \hspace{-1cm}\times \sum_{i=0}^{n-1}(-1)^i
\binom{n-1}{i} \varphi_{1}(1 - i \delta) .
\end{eqnarray}
The function $\varphi_{1}(1 - i \delta)$ can be evaluated as
\begin{eqnarray} \label{eq:varPhi1_closedform}
\varphi_{1}(1 - i \delta) &=& \int_{0}^{\delta} (1- i \delta -
t_1)^{M-2} {\rm d} t_1 \nonumber\\ &=& \frac{(1- i \delta)^{M-1} -
\left(1- (i+1) \delta\right)^{M-1}}{M-1} . ~~~~~
\end{eqnarray}
Substituting (\ref{eq:varPhi1_closedform}) into
(\ref{eq:varPhi_Induct}) yields a closed-form solution, which we
simplify as follows: \begin{eqnarray} \varphi_{n}(1) & =&
\frac{\Gamma(M-n)}{\Gamma(M)} \sum_{i=0}^{n-1}(-1)^i \binom{n-1}{i}
\nonumber\\ &&\times  \left( (1- i\delta)^{M-1} -
[1- (i+1) \delta]^{M-1} \right) \nonumber\\
&=&  \frac{\Gamma(M-n)}{\Gamma(M)} \sum_{i=0}^{n} \binom{n}{i}
(-1)^{i} (1-i \delta)^{M-1} \nonumber\\
&=& \frac{\Gamma(M-n)}{\Gamma(M)} \sum_{k=0}^{M-1}  \binom{M-1}{k}
(-1)^{k} \nonumber\\ &&\times  \bigg[ \sum_{i=0}^{n} \binom{n}{i}
(-1)^{i} i^k \bigg]\delta^k .
\end{eqnarray}
Since \cite{Gradshteyn2000}
\begin{eqnarray}\label{eq:iden1}
\sum_{k = 0}^{N} \binom{N}{k} (-1)^{k} k^{(n-1)} = 0 , ~~~ 1\leq n
\leq N , \end{eqnarray} \begin{eqnarray}\sum_{k = 0}^{N}
\binom{N}{k} (-1)^{k} k^N = (-1)^N N!, ~~~ N \geq 0,
\end{eqnarray}
we obtain $\varphi_{n}(1) = \frac{\Gamma(M-n)}{\Gamma(M)}
\sum_{k=n}^{M-1} \binom{M-1}{k} (-1)^{k} \left[\sum_{i=0}^{n}
\binom{n}{i} (-1)^{i} i^k \right] \delta^k$.
Substituting
into (\ref{eq:prob_sel_n}) yields (\ref{eq:mu_delta}).

\section{Proof of \emph{Lemma} \ref{lemma:distri_beta}}
\label{proof_lemma_distri_beta} Similar to the proof of \emph{Lemma}
\ref{lemma:distri_mu_delta}, we assume $\mathbf{q}_i = \mathbf{e}_i$
without loss of generality. Then the numerator of
(\ref{eq:beta_dist}) is given by
\begin{align}\label{eq:prob_numerator}
& \text{Pr} \left( \sum_{i=1}^{n-1} |\mathbf{v}_k^H \mathbf{q}_i^H |^2
\leq 1- x , |\mathbf{v}_k^H \mathbf{q}_1^H |^2< \delta,\cdots,
|\mathbf{v}_k^H \mathbf{q}_{n-1}^H |^2 < \delta \right) \nonumber \\
& \hspace*{1cm} = \text{Pr} \left( \sum_{i=1}^{n-1} |v_i|^2 \leq
1-x, |v_1|^2< \delta,\cdots, |v_{n-1}|^2< \delta \right) .
\end{align}
Recalling that $t_i = |v_i|^2$, $i=1,2,\cdots, n-1$, we can evaluate
(\ref{eq:prob_numerator}) using the joint p.d.f.\ $f(t_1, \ldots,
t_{n-1})$ given in (\ref{eq:joint_density}) in Appendix
\ref{proof_lemma:distri_mu_delta}. For $n=2$, we have
\begin{eqnarray} \label{eq:n2}
\text{Pr}\big(  |\mathbf{v}_k^H \mathbf{q}_1^H |^2 \leq 1 - x,
|\mathbf{v}_k^H \mathbf{q}_1^H |^2< \delta\big) &=& \left\{
\begin{array}{lll}
\int_{0}^{\delta} (M-1) \left(1 - {t_1}\right)^{M-2} {\rm d}t_1  ~~ &  x \leq 1- \delta \\
\int_{0}^{1-x} (M-1) \left(1 - {t_1} \right)^{M-2} {\rm d}t_1 ~~ &   1- \delta < x \leq 1\\
0 ~~ & x > 1
\end{array}
\right.
\end{eqnarray}
Solving the integrals in (\ref{eq:n2}) and combining the result with
(\ref{eq:mu_delta}) and (\ref{eq:beta_dist}) leads to the explicit
solution given in (\ref{eq:distri_beta_2}).
For $n > 2$, the problem is much more difficult. In this case, using
(\ref{eq:joint_density}), we obtain
\begin{align}
&\text{Pr} \left( \sum_{i=1}^{n-1} |\mathbf{v}_k^H \mathbf{q}_i^H
|^2 \leq 1- x , |\mathbf{v}_k^H \mathbf{q}_1^H |^2< \delta,\cdots,
|\mathbf{v}_k^H \mathbf{q}_{n-1}^H |^2 < \delta \right) \nonumber \\
& \hspace*{1cm} = \left\{
\begin{array}{lll}
0 ~~ & x > 1 \\
\mu_{n} (\delta)
~~ & x \leq 1-(n-1)\delta\\
\frac{\Gamma(M)}{\Gamma(M-n+1)} \int_{t_{n-1}}\cdots \int_{t_1}
\left(1-\sum_{i=1}^{n-1}t_i \right)^{M-n} {\rm d}t_1\cdots {\rm
d}t_{n-1}
~~ &   1- (n-1)\delta < x \leq 1\\
\end{array}
\right. \label{eq:FirstEq}
\end{align}
with the integration region for the remaining multi-dimensional
integral defined in the lemma statement. Combining
(\ref{eq:FirstEq}) with (\ref{eq:mu_delta}) and (\ref{eq:beta_dist})
leads to (\ref{eq:distri_beta_n}).

\section{Proof of \emph{Lemma} \ref{lemma:distri_beta_bound}}
\label{proof_lemma:distri_beta_bound} We can upper bound the c.d.f.\
(\ref{eq:distri_beta_n}), for $n \geq 2, 1- (n-1)\delta< x \leq 1$,
as follows
\begin{eqnarray}\label{eq:upper_beta_k}
F_{\beta (n)}(x) &\leq & 1 - \frac{\Gamma(M)}{\Gamma(M-n+1)
\mu_{n}(\delta)} \nonumber\\ && \hspace{-2cm} \times
\int_{0}^{\frac{1-x}{n-1}}\cdots \int_{0}^{\frac{1- x}{n-1}} \left(
1-\sum_{i=1}^{n-1}t_i \right)^{M-n} {\rm
d}t_1\cdots {\rm d}t_{n-1} \nonumber\\
&=& 1- \frac{\mu_{n}\left(\frac{1-x}{n-1}\right)}{\mu_{n}(\delta)}
\end{eqnarray}
where the second line follows from (\ref{eq:prob_sel_n}).
For $n = 2$, we have
\begin{align}\label{eq:upper_n=2}
F_{\beta (2)}(x) \leq 1 -\frac{\mu_{2}\left(1-x \right)}{\mu_{2}(\delta)} = \frac{x^{M-1}
- (1-\delta)^{M-1}}{(1-\delta)^{M-1}}
\end{align}
which is exactly the right-hand side of (\ref{eq:distri_beta_2}).

We can establish the corresponding lower bound via
\begin{eqnarray}\label{eq:lower_beta_k}
F_{\beta (n)}(x) &\geq &  1 - \frac{\Gamma(M)}{\Gamma(M-n+1)
~\mu_{n} (\delta)}\nonumber\\ && \hspace{-1cm} \times \mathop {\int
{ \cdots \int {} } }\limits_{\scriptstyle \sum_{i=1}^{n-1} t_i \leq
1-x \hfill \atop \scriptstyle t_1 \geq 0, \cdots, t_{n-1} \geq 0
\hfill} \left(1- \sum_{i=1}^{n-1} t_i \right)^{M-n} {\rm
d}t_1\cdots {\rm d}t_{n-1} \nonumber\\
&=& 1 - \frac{\Gamma(M)}{\Gamma(M-n+1) ~\mu_{n} (\delta)}
\nonumber\\ && \times \int_{0}^{1-x} (1-y)^{M-n}
\frac{y^{n-2}}{(n-2)!} ~{\rm d}
y\nonumber\\
&=& 1 - \frac{I_{1-x} (n-1, M-n+1)}{\mu_{n} (\delta)},
\end{eqnarray}
where we have used the identity\cite{Gradshteyn2000} $ \mathop {\int
\int {  \cdots \int {} } }\limits_{\scriptstyle \sum_{i=1}^{n}t_i
\leq h \hfill \atop \scriptstyle t_1 \geq 0, \cdots, t_{n} \geq 0
\hfill} ~{\rm d} t_1 \cdots ~{\rm d} t_n = \frac{h^n}{n!}$.  For
$n=2$, it is easily verified that (\ref{eq:lower_beta_k}) is equal
to (\ref{eq:upper_n=2}).

\section{Proof of \emph{Lemma} \ref{lemma:distri_DPC}}\label{proof_distri_DPC}
Recalling that for uncorrelated Wishart matrices, the eigenvalues and their corresponding eigenvectors
are independent, it follows that $\lambda_{k,\text{max}}$ is independent of $\beta_{k}(n)$, $\tilde{\beta}_{k}(n)$, and $\bar{\beta}_{k}(n)$. Thus, the c.d.f.s of ${\gamma}_{k}(n)$, $\tilde{\gamma}_{k} (n)$, and $\bar{\gamma}_{k} (n)$, can be derived as
$F_{{\gamma} (n)} (x) = \int_{0}^{\infty}
F_{{\beta}(n)}  ( x/ y ) f_{\text{max}} (y)
{\rm d} y$, $F_{\tilde{\gamma} (n)} (x) = \int_{0}^{\infty}
F_{\tilde{\beta}(n)}  ( x / y ) f_{\text{max}} (y)
{\rm d} y$, and $F_{\bar{\gamma} (n)} (x) = \int_{0}^{\infty}
F_{\bar{\beta}(n)}  ( x / y ) f_{\text{max}} (y)
{\rm d} y$ respectively,
where $f_{\text{max}} (\cdot)$ is the p.d.f. of the maximum eigenvalue of
$\mathbf{H}_k \mathbf{H}_k^H$. Together with \emph{Lemma} \ref{lemma:distri_beta_bound}, it follows trivially that
$F_{\bar{\gamma} (n)} (x) \leq F_{\gamma (n)} (x) \leq F_{\tilde{\gamma} (n)} (x)$, where the equalities hold for $n=2$.

What remains is to derive closed-form expressions for
$F_{\tilde{\gamma} (n)} (x)$ and $F_{\bar{\gamma} (n)} (x)$. First
consider $F_{\tilde{\gamma} (n)} (x)$.  Recalling
 (\ref{eq:dom_distri_beta_n}), and noting that for $ 1- (n-1)\delta <
x \leq 1$, $F_{\tilde{\beta}(n)}(x)$ can be re-expressed using
(\ref{eq:mu_delta}) as
\begin{align}\label{eq:beta_cdf}
F_{\tilde{\beta}(n)}(x) &= 1-
\frac{1}{\mu_{n}(\delta)}\sum_{k=n-1}^{M-1} \binom{M-1}{k}(-1)^k
\nonumber\\ & \hspace{-1cm} \times  \bigg[ \sum_{i=0}^{n-1}
\binom{n-1}{i}(-1)^i \left(\frac{i}{n-1} \right)^k (1-x)^k \bigg]
\end{align}
it follows using \emph{Lemma} \ref{lemma:max_eigen_pdf} that
\begin{align}\label{eq:distri_temp}
F_{\tilde{\gamma} (n)} (x) & = F_{\text{max}} \left( \frac{x}{t}
\right) - \frac{1}{\mu_{n}(\delta)} \sum_{k=n-1}^{M-1}
\binom{M-1}{k}(-1)^k \nonumber\\
& \times\bigg[ \sum_{i=0}^{n-1} \binom{n-1}{i}(-1)^i \left(
\frac{i}{n-1} \right)^k \bigg]  \sum_{r=1}^{p}
\sum_{s=q-p}^{(N+M-2r)r} \nonumber\\
&   a_{r,s} \int_{x}^{\frac{x}{t}} \left( 1-\frac{x}{y} \right)^k
y^s e^{-ry} {\rm d} y.
\end{align}
By applying the transformation $z= \frac{y}{x}$ along with some
elementary algebraic manipulations, the remaining integral is
evaluated as
\begin{align}\label{eq:simplification1}
&\int_{x}^{\frac{x}{t}} \left( 1-\frac{x}{y} \right)^k \frac{y^s}{
e^{ry}} {\rm d} y \nonumber\\& \hspace{0.5cm}=
\sum_{j=0}^{k} \binom{k}{j} (-1)^{ k-j } r^{k-j-s-1} x^{k-j}
\nonumber\\ &  \hspace{0.5cm} \times \left[ \Gamma(j-k+s+1, rx) -
\Gamma\left(j-k+s+1, \frac{rx}{t}\right) \right]. \nonumber
\end{align}
Substituting this expression into (\ref{eq:distri_temp}), we readily
obtain the result (\ref{eq:upper_distri_DPC}). A closed-form
expression for $F_{\bar{\gamma} (n)} (x)$ can be obtained in a
similar manner, and is omitted due to space limitations.


\section{Asymptotic expansion of c.d.f.s of $\tilde{\gamma}_{k}(n)$ and $\bar{\gamma}_{k}(n)$
for large $x$} \label{app:tail_gamman}

First note that the tail behavior (large $x$) of $F_{\text{max}} (x)$ is given by \cite{Alireza08}
\begin{eqnarray}\label{eq:tail_eig}
F_{\text{max}} (x) =1- \frac{e^{-x} x^{M+N-2}}{\Gamma(M)\Gamma(N)} +
O (e^{-x} x^{M+N-3}).
\end{eqnarray}
Then, the corresponding expansion for the term $F_{\text{max}}
(\frac{x}{t})$ in both (\ref{eq:upper_distri_DPC}) and
(\ref{eq:lower_distri_DPC}) follows immediately.
In the following, we require a corresponding expansion for the
remaining terms in (\ref{eq:upper_distri_DPC}) and
(\ref{eq:lower_distri_DPC}). First consider
(\ref{eq:upper_distri_DPC}). Since the remaining terms in this case
involve the upper incomplete gamma function $\Gamma(n,x)$,
we require an asymptotic expansion for $\Gamma(n,x)$ at $x
\rightarrow \infty$. Using the definition and integrating by parts,
for large $x$ we have $\Gamma(n,x) = e^{-x} x^{n-1} [ 1 +
\frac{n-1}{x} + \frac{(n-1)(n-2)}{x^2} +\cdots]$. Since $t <1$, the
terms that decay most slowly in the summation in
(\ref{eq:upper_distri_DPC}) can be expressed as
\begin{eqnarray}\label{eq:expansion1}
\mathcal{J}_1
& = \sum_{k=n-1}^{M-1} \mathcal{C}_k \sum_{s=q-p}^{N+M-2}
\frac{a_{1,s} x^s}{e^{x}} \sum_{j=0}^{k} \frac{ \binom{k}{j}}{
(-1)^{ k-j }}\nonumber\\
& \times \bigg[ 1+ \frac{j-k+s}{x} + \frac{(j-k+s)(j-k+s-1)}{x^2} +
\cdots \bigg],
\end{eqnarray}
where 
\begin{align}
\mathcal{C}_k = \binom{M-1}{k}(-1)^k \bigg[
\sum_{i=0}^{n-1} \binom{n-1}{i}(-1)^i \left(\frac{i}{n-1} \right)^k \bigg] \; .
\end{align}

Using (\ref{eq:iden1}) we can obtain
\begin{align}
\label{eq:Ident2} \sum_{ j= 0}^{k} \binom{k}{j} (-1)^{k-j} j^{(m-1)}
&= 0 , \; \; \; \; 1 \leq m \leq k,
\nonumber\\
\sum_{ j= 0}^{k} \binom{k}{j} (-1)^{k-j} j^{k} &= k ! , \; \; \; \;
k \geq 1,
\end{align}
from which it follows that in (\ref{eq:expansion1}), $\sum_{j=0}^{k} \binom{k}{j} (-1)^{ k-j }
\frac{\prod_{v=1}^{m}(j-k+s+1-v)}{x^m} = 0$ for $1 \leq m
< k-1$, and also that $\sum_{j=0}^{k} \binom{k}{j} (-1)^{ k-j }
\frac{\prod_{v=1}^{k}(j-k+s+1-v)}{x^k} = \frac{ k!}{x^k}$.
We then have
\begin{align}
\mathcal{J}_1 = \sum_{k=n-1}^{M-1} \mathcal{C}_k
\sum_{s=q-p}^{N+M-2} \frac{ a_{1,s} x^s  k!}{e^{x}} \left(
\frac{1}{x^k} + O\left(\frac{1}{x^{k+1}} \right) \right) ,
\end{align}
which upon substituting for $\mathcal{C}_k$ and applying some manipulations using (\ref{eq:Ident2}) gives
\begin{eqnarray}
\mathcal{J}_1
&=& \frac{(M-1)! (n-1)!}{(M-n)!(n-1)^{n-1}} a_{1,M+N-2} e^{-x}
x^{M+N-n-1} \nonumber \\ && + ~O(e^{-x}x^{M+N-n-2}) \; .
\label{eq:expansion2}
\end{eqnarray}
From (\ref{eq:tail_eig}), we have $f_{\text{max}} (x) = \frac{e^{-x}
x^{N+M-2}}{\Gamma(M)\Gamma(N)} + O({e^{-x} x^{N+M-3}})$. Therefore
$a_{1,N+M-2}= \frac{1}{\Gamma(M)\Gamma(N)}$. Together with
(\ref{eq:expansion2}) and (\ref{eq:tail_eig}), we have
(\ref{eq:upper_tail_distri}).  By using a similar method, the terms
that decay most slowly in the summation in
(\ref{eq:lower_distri_DPC}) can be obtained. That result, used with
(\ref{eq:tail_eig}), yields (\ref{eq:lower_tail_distri}).

\section{Proof of \emph{Lemma}
\ref{lemma:SINR_bound_ZFDPC1}}\label{proof_SINR_bound_ZFDPC1}

Recall that $F_{\bar{\gamma}_{\pi(n)}}(x) \leq F_{{\gamma}_{\pi(n)}}(x)
\leq F_{\tilde{\gamma}_{\pi(n)}}(x)$. For $\gamma_{\pi(n)},
n \in \{ 2,\cdots, M \}$, and large $K$, with (\ref{eq:extre_gamma_til}),
$\text{Pr}\{ u_n - \log\log \sqrt K \leq \gamma_{\pi(n)}\} \geq
\text{Pr}\{ u_n - \log\log \sqrt K \leq \tilde{\gamma}_{\pi(n)}\}
\geq 1- O\bigg( \frac{1}{\log K} \bigg)$. Similarly, with (\ref{eq:extre_gamma_bar}) we have
$\text{Pr}\{  \gamma_{\pi(n)} \leq \chi_n + \log\log \sqrt K \} \geq
\text{Pr}\{ \bar{\gamma}_{\pi(n)} \leq  \chi_n + \log\log \sqrt K \}
\geq 1- O\bigg( \frac{1}{\log K} \bigg)$. Thus,
\begin{eqnarray}\label{eq:sinr_up_low}
&& \hspace{-1cm}\text{Pr}\{ u_n - \log\log \sqrt K \leq
\gamma_{\pi(n)} \leq \chi_n
+ \log\log \sqrt K  \} \nonumber\\
&& \geq 1 - O\bigg( \frac{1}{\log K} \bigg) .
\end{eqnarray}
For $n = 1$, the asymptotic distribution of
$\gamma_{\pi(n)}$ has been characterized in \cite{Mohammad08}. Using that result, along with (\ref{eq:sinr_up_low}),
the lemma follows upon noting that $\zeta_{\pi(n)} = \rho \gamma_{\pi(n)}$.

\section{Proof of \emph{Theorem}
\ref{theorem:sum_rate_GZFDPC}}\label{proof_theorem:sum_rate_GZFDPC}

Using (\ref{eq:sinr_up_low2}) we can obtain $ \text{Pr}\bigg\{
\frac{\log_2 (1 +  \varpi_n - \rho \log\log\sqrt{K})}{ \log_2{[\rho
\log K]} }  \leq \frac{ \log_2 (1+ \zeta_{\pi(n)})}{\log_2[\rho\log
K]} \leq \frac{\log_2 (1 + \upsilon_{n} + \rho\log\log\sqrt{K})}{
\log_2{[\rho \log K]} }  \bigg\} \geq 1- O \bigg(\frac{1}{\log K}
\bigg)$. Substituting (\ref{eq:varpi_n}) and (\ref{eq:upsilon_n})
and letting $K \rightarrow \infty$, the left-hand side and
right-hand side inequality within $\text{Pr} \{\cdot \}$ converge to
the same value. Thus, $\lim_{K \rightarrow \infty}  \frac{ \log_2
(1+ \zeta_{\pi(n)})}{\log_2[\rho\log K]} = 1 $ with probability 1,
and (\ref{eq:converge_1}) holds. To establish (\ref{eq:converge_2}),
we employ the following upper bound on
$\mathcal{E}\{R_{\text{BC}}\}$ derived in \cite{Sharif05}:
\begin{align}\label{eq:sum_rate_BC}
\mathcal{E} \{R_{\text{BC}}\} \leq M \log_2 \big( 1 + \rho (\log K +
O(\log \log K)) \big).
\end{align}
From \emph{Lemma} \ref{lemma:SINR_bound_ZFDPC1}, we have
$\text{Pr}\bigg\{ \log_2 (1+ \zeta_{\pi(n)}) \geq \log_2 (1 +
\varpi_{n} - \rho\log\log\sqrt{K}) \bigg\} \geq 1- O
\bigg(\frac{1}{\log K} \bigg)$.
Thus,
\begin{eqnarray}\label{eq:order_difference_GZFDPC}
&& \hspace{-1cm}\mathcal{E} \{ R_{\text{BC}}\} - \mathcal{E} \{
R_{\text{ZFDPC-SUS}} \} \nonumber\\
& \leq & M \log \big( 1 + \rho (\log K + O(\log \log K)) \big) \nonumber \\
&& - \bigg( 1 - O \bigg(\frac{1}{\log K } \bigg) \bigg)
\nonumber\\&& \times \sum_{n=1}^{M} \log \big(1 + \varpi_{n} -
\rho\log\log\sqrt{K}
\big) \nonumber\\
& \sim & \sum_{n=1}^{M} \log\bigg(1 + \frac{O(\log \log K)}{1 +
\varpi_{n} - \rho\log\log\sqrt{K}}\bigg) \nonumber\\&& + ~ O
\bigg(\frac{1}{\log K }  \bigg) M~ O(\log \log K ) \nonumber\\
& \sim & O \bigg(\frac{\log \log K}{\log K} \bigg)
\end{eqnarray}
where we have used $\log(1+x) \approx x$ for $x\ll 1$, and $x\sim y$
means $\lim_{K \rightarrow \infty} x/y = 1$.

\section{Proof of \emph{Theorem}
\ref{theorem:sum_rate_GZFBF}}\label{proof_theorem:sum_rate_ZF}

From \cite{Yoo06}, for small enough $\delta$,
$\varrho_{\pi(n)} > \frac{\gamma_{\pi(n)}}{1+ e(\delta)}$, where
$e(\delta) = \frac{(M-1)^4 \delta}{1-(M-1)\delta} $. Using this
result, together with (\ref{eq:sum_rate_BC}) and
(\ref{eq:sinr_up_low2}), and following a similar method as in
Appendix \ref{proof_theorem:sum_rate_GZFDPC}, we have
\begin{eqnarray}\label{eq:zf_difference}
\mathcal{E} \{R_{\text{BC}}\}& -& \mathcal{E}
\{R_{\text{ZFBF-SUS}}\}
\nonumber \\
& \leq & M \log \big(1+ \rho(\log K + O(\log \log K)) \big) -
\mathcal{E} \bigg\{ \sum_{n=1}^{M} \log\bigg(1+
\frac{\rho\gamma_{\pi
(n)}}{1+ e(\delta)}\bigg) \bigg\}\nonumber\\
& \leq & M \log \big(1+ \rho (\log K + O(\log \log K)) \big) -
\sum_{n=1}^{M} \bigg(1- O\bigg(\frac{1}{\log K}\bigg)\bigg)
\log\bigg(1+ \frac{\varpi_n - \rho \log\log \sqrt K}{1+e(\delta)}
\bigg)\nonumber\\
&\sim& \sum_{n=1}^{M} \log \bigg(1+ \frac{\rho\big( e(\delta) \log K
+O(\log\log K) \big)}{1+ ~\big(\varpi_n - \rho \log\log \sqrt K\big)
\sum_{i=0}^{\infty} \big(- e(\delta)\big)^i }
\bigg) + O\bigg(\frac{\log\log K}{\log K}\bigg)\nonumber\\
&\sim & M e(\delta) +O\bigg(\frac{\log\log K}{\log K}\bigg) ,
\end{eqnarray}
where we have used the fact that for small enough $\delta$,
$|e(\delta)|< 1$, thus $\frac{1}{1 + e(\delta)} =
\sum_{i=0}^{\infty} \big(- e(\delta)\big)^i$. So we can see that as
long as $e(\delta) \sim o(1)$, or equivalently $\delta \sim o(1)$,
whilst satisfying the conditions in (\ref{eq:Conditions}), the
difference will become zero as $K \rightarrow \infty$.
However, obviously ZFBF-SUS with a smaller candidate set at each
iteration (i.e., reduced $| \mathcal{U}_n |$) can not achieve more
sum rate than ZFBF-SUS with a larger candidate set at each
iteration. Thus, with larger $\delta$, there will be more candidate
users for each iteration and the average sum rate will increase, or
at least maintain. So the condition $\delta \sim o(1)$ can be
ignored, thereby establishing (\ref{eq:SumRate_Diff}). From
(\ref{eq:zf_difference}), the difference in sum rate is at most $O
\big(\frac{\log\log K}{\log K} \big)$.

\bibliographystyle{IEEEtran}
\bibliography{sunliang_bib}

\begin{thebibliography}{10}
\providecommand{\url}[1]{#1}
\csname url@rmstyle\endcsname
\providecommand{\newblock}{\relax}
\providecommand{\bibinfo}[2]{#2}
\providecommand\BIBentrySTDinterwordspacing{\spaceskip=0pt\relax}
\providecommand\BIBentryALTinterwordstretchfactor{4}
\providecommand\BIBentryALTinterwordspacing{\spaceskip=\fontdimen2\font plus
\BIBentryALTinterwordstretchfactor\fontdimen3\font minus
  \fontdimen4\font\relax}
\providecommand\BIBforeignlanguage[2]{{%
\expandafter\ifx\csname l@#1\endcsname\relax
\typeout{** WARNING: IEEEtran.bst: No hyphenation pattern has been}%
\typeout{** loaded for the language `#1'. Using the pattern for}%
\typeout{** the default language instead.}%
\else
\language=\csname l@#1\endcsname
\fi
#2}}

\bibitem{W_Yu02}
W.~Yu and J.~M. Cioffi, ``{Sum capacity of a Gaussian vector broadcast
  channels},'' \emph{IEEE Trans. Inform. Theory}, vol.~50, no.~9, pp.
  1875--1892, Sep. 2002.

\bibitem{Viswanath03}
P.~Viswanath and D.~N.~C. Tse, ``{Sum capacity of the vector Gaussian broadcast
  channel and uplink-downlink duality},'' \emph{IEEE Trans. Inform. Theory},
  vol.~49, no.~8, pp. 1912--1921, Aug. 2003.

\bibitem{Weingarten06}
H.~Weingarten, Y.~Steinberg, and S.~Shamai{ }(Shitz), ``{The capacity region of
  the Gaussian multiple-input multiple-output broadcast channel},''
  \emph{{IEEE} Trans. Inform. Theory}, vol.~52, no.~9, pp. 3936--3964, Sep.
  2006.

\bibitem{Vishwanath03}
S.~Vishwanath, N.~Jindal, and A.~Goldsmith, ``{Duality, achievable rates, and
  sum-rate capacity of Gaussian MIMO broadcast channels},'' \emph{IEEE Trans.
  Inform. Theory}, vol.~49, no.~10, pp. 2658--2668, Oct. 2003.

\bibitem{Caire03}
G.~Caire and S.~Shamai{ }(Shitz), ``{On the achievable throughput of a
  multi-antenna Gaussian broadcast channel},'' \emph{IEEE Trans. Inform.
  Theory}, vol.~49, no.~7, pp. 1691--1706, Jul. 2003.

\bibitem{Love07}
A.~D. Dabbagh and D.~J. Love, ``{Precoding for multiple antenna Gaussian
  broadcast channels with successive zero-forcing},'' \emph{IEEE Trans. Signal
  Process.}, vol.~55, no.~7, pp. 3837--3850, Jul. 2007.

\bibitem{Peel05}
C.~B. Peel, B.~M. Hochwald, and A.~L. Swindlehurst, ``{A vector-perturbation
  technique for near-capacity multiantenna multi-user communication - Part I:
  Channel inversion and regularization},'' \emph{IEEE Trans. Commun.}, vol.~53,
  no.~1, pp. 195--202, Jan. 2005.

\bibitem{Hochwald05}
B.~M. Hochwald, C.~B. Peel, and A.~L. Swindlehurst, ``{A vector-perturbation
  technique for near-capacity multiantenna multiuser communication - Part II:
  Perturbation},'' \emph{IEEE Trans. Commun.}, vol.~53, no.~3, pp. 537--544,
  Mar. 2005.

\bibitem{Spencer04}
Q.~H. Spencer, A.~L. Swindlehurst, and M.~Haardt, ``{Zero-forcing methods for
  downlink spatial multiplexing in multiuser MIMO channels},'' \emph{IEEE
  Trans. Signal Process.}, vol.~52, no.~2, pp. 461--471, Feb. 2004.

\bibitem{ZTu03}
Z.~Tu and R.~S. Blum, ``{Multiuser diversity for a dirty paper approach},''
  \emph{IEEE Commun. Lett.}, vol.~7, no.~8, pp. 370--372, Aug. 2003.

\bibitem{Dimic05}
G.~Dimic and N.~Sidiropoulos, ``{On the downlink beamforming with greedy user
  selection: Performance analysis and a simple new algorithm},'' \emph{{IEEE}
  Trans. Signal Process.}, vol.~53, no.~10, pp. 3857--3868, Jul. 2005.

\bibitem{Shen06}
Z.~Shen, R.~Chen, J.~G. Andrews, R.~W. Heath{ }Jr., and B.~L. Evans, ``{Low
  complexity user selection algorithms for multiuser MIMO systems with block
  diagonalization},'' \emph{IEEE Trans. Signal Process.}, vol.~54, no.~9, pp.
  3658--3663, Sep. 2006.

\bibitem{Yoo06}
T.~Yoo and A.~J. Goldsmith, ``{On the optimality of multi-antenna broadcast
  scheduling using zero-forcing beamforming},'' \emph{{IEEE} J. Sel. Areas
  Commun.}, vol.~24, no.~3, pp. 528--541, Mar. 2006.

\bibitem{Mohammad08}
M.~A. Maddah-Ali, M.~Ansari, and A.~K. Khandani, ``{Broadcast in MIMO systems
  based on a generalized QR decomposition: signaling and performance
  analysis},'' \emph{IEEE Trans. Inform. Theory}, vol.~54, no.~3, pp.
  1124--1138, Mar. 2008.

\bibitem{Alireza08}
A.~Bayesteh and A.~K. Khandani, ``{On the user selection for MIMO broadcast
  channels},'' \emph{IEEE Trans. Inform. Theory}, vol.~54, no.~3, pp.
  1086--1107, Mar. 2008.

\bibitem{Sharif05}
M.~Sharif and B.~Hassibi, ``{On the capacity of {MIMO} broadcast channels with
  partial side information},'' \emph{IEEE Trans. Inform. Theory}, vol.~2,
  no.~21, pp. 506--522, Feb. 2005.

\bibitem{Sharif07}
------, ``{A comparison of time-sharing, DPC, and beamforming for MIMO
  broadcast channels with many users},'' \emph{IEEE Trans. Commun.}, vol.~55,
  no.~1, pp. 11--15, Jan. 2007.

\bibitem{Love08}
J.~Wang, D.~J. Love, and M.~D. Zoltowski, ``{User selection with zero-forcing
  beamforming achieves the asymptotically optimal sum rate},'' \emph{IEEE
  Trans. Signal Process.}, vol.~56, no.~8, pp. 3713--3726, Aug. 2008.

\bibitem{H_Viswanathan03}
H.~Viswanathan, S.~Venkatesan, and H.~Huang, ``{Downlink capacity evaluation of
  cellular networks with known-interference cancellation},'' \emph{{IEEE} J.
  Sel. Areas Commun.}, vol.~21, no.~6, pp. 802--811, Jun. 2003.

\bibitem{Mallik03}
P.~A. Dighe, R.~K. Mallik, and S.~S. Jamuar, ``{Analysis of transmit-receive
  diversity in Rayleigh fading},'' \emph{IEEE Trans. Commun.}, vol.~51, no.~4,
  pp. 694--703, Apr. 2003.

\bibitem{Maaref05}
A.~Maaref and S.~A{\"{\i}}ssa, ``{Closed-form expressions for the outage and
  ergodic Shannon capacity of MIMO MRC systems},'' \emph{IEEE Trans. Commun.},
  vol.~53, no.~7, pp. 1092--1095, Jul. 2005.

\bibitem{oreder_statistics03}
H.~David and H.~Nagaraja, \emph{{Order Statistics}}, 3rd~ed.\hskip 1em plus
  0.5em minus 0.4em\relax New York: John Wiley and Sons, 2003.

\bibitem{XZhang03}
X.~Zhang and S.-Y. Kung, ``{Capacity analysis for parallel and sequential MIMO
  equalizers},'' \emph{{IEEE} Trans. Signal Process.}, vol.~11, no.~51, pp.
  2989--3002, Nov. 2003.

\bibitem{Jindal05}
N.~Jindal, W.~Rhee, S.~Vishwanath, S.~Jafar, and A.~Goldsmith, ``{Sum power
  iterative water-filling for multi-antenna Gaussian broadcast channels},''
  \emph{IEEE Trans. Inform. Theory}, vol.~51, no.~4, pp. 1570--1580, Apr. 2005.

\bibitem{Wang_report07}
J.~Wang, D.~J. Love, and M.~D. Zoltowski, \emph{{\emph{A Result on Order
  Statistics}}}.\hskip 1em plus 0.5em minus 0.4em\relax [Online]. Available:
  http:// docs.lib.purdue.edu/ecetr/347, Tech. Rep., Purdue Univ.,West
  Lafayette, IN, 2007.

\bibitem{Kendall1961}
M.~G. Kendall, \emph{{A course in the geometry of n dimensions}}, 1st~ed.\hskip
  1em plus 0.5em minus 0.4em\relax London, U.K.: Charles Griffin Co., Ltd.,
  1961.

\bibitem{Gradshteyn2000}
I.~S. Gradshteyn and I.~M. Ryzhik, \emph{{Table of Integrals, Series, and
  Products}}, 6th~ed.\hskip 1em plus 0.5em minus 0.4em\relax New York:
  Academic, 2000.

\end{thebibliography}

\end{document}